\documentclass{article}

\usepackage{graphics}
\usepackage{epstopdf}
\usepackage{tikz}
\usepackage{pgf}
\usepackage{dutchcal}
  \usepackage{color}
  \definecolor{purple}{rgb}{0.5,0,1}
  \definecolor{orange}{cmyk}{0,0.7,1,0}
  \definecolor{midgrey}{gray}{0.5}
  \usepackage{amssymb}
    \usepackage{amsthm}
  \usepackage{amsmath}

  \newtheorem{theorem}{Theorem}[section]
  \newtheorem{remark}{Remark}[section]
  \newtheorem{corollary}{Corollary}[section]
    \newtheorem{lemma}{Lemma}[section]

  \makeatletter
  \let\c@equation\undefined
  \let\c@section\undefined
  \let\c@subsection\undefined
  \let\c@zad\undefined
  \makeatother
  \newcounter{section}
    \def\thesection{\arabic{section}}
\newcommand{\mr}[1]{\mathring{#1}}

  \newcounter{equation}
  \def\theequation{\thesection.\arabic{equation}}
  \newcounter{subsection}[section]

\newcommand{\beq}{\begin{equation}}
\newcommand{\eeq}{\end{equation}}

\newcommand{\papap}{\papa \end{proof}}

\newcommand{\ov}[1]{\overline{#1}}
\newfont{\smoldita}{cmmib8}
\newfont{\boldita}{cmmib10}
\newfont{\bboldita}{cmmib10}

\newcommand{\nn}{\nonumber}
\newcommand{\e}{\epsilon}

\newcommand{\p}{\partial}

\newcommand{\cl}[2]{\int\limits_{#1}^{#2}}

\newcommand{\ti}[1]{\tilde{#1}}

\newcommand{\la}{\lambda}

\newfont{\llmt}{cmmib10 scaled\magstep2}
\newfont{\lmt}{cmmib10 scaled\magstep1}
\newfont{\mt}{cmmib10}
\newfont{\smt}{cmmib8}
\newfont{\bgr}{cmmib9 scaled\magstep1}
\newcommand{\mb}[1]{\boldsymbol{#1}}
\newcommand{\mbb}[1]{\mathbb{#1}}
\newcommand{\mc}[1]{\mathcal{#1}}

\begin{document}
\title{Multiscale malaria models and their uniform in-time asymptotic analysis\footnote{Both authors was supported by the DSI/NRF SARChI M3B2 grant N 82770 }}
\author{J. Banasiak$^{1,2,}$\footnote{Corresponding author, jacek.banasiak@up.ac.za} and S.Y. Tchoumi$^{1,3}$}
\date{\small{1. Department of Mathematics and Applied Mathematics, University of Pretoria, Pretoria, South Africa\\
2. Institute of Mathematics, Łódź University of Technology, Łódź, Poland\\
3. Department of Mathematics and Computer Sciences ENSAI, University of Ngaoundere, Ngaoundere, Cameroon}}

\maketitle
\begin{abstract}
In this paper, we show that an extension of the classical Tikhonov--Fenichel asymptotic procedure applied to multiscale models of vector-borne diseases, with time scales determined by the dynamics of human and vector populations, yields a simplified model approximating the original one in a consistent, and uniform for large times, way. Furthermore, we construct a higher-order approximation based on the classical Chapman--Enskog procedure of kinetic theory and show, in particular, that it is equivalent to the dynamics on the first-order approximation of the slow manifold in the Fenichel theory.
\end{abstract}
\section{Introduction}
Infectious diseases are driven by a complex interplay of many mechanisms, often acting at vastly different time or space scales. For example, even ignoring the spatial dependence, the evolution of malaria, which is the main subject of the presented paper, is driven by interactions between humans and female \textit{Anopheles} mosquitoes (with their vital dynamics), and the parasite \textit{Plasmodium}, see, e.g., \cite{Chitnis, Ngwa}. Even ignoring the dynamics of the parasite, the interactions of mosquitoes and humans coupled with their vital dynamics and ecology lead to formidable models consisting of many highly coupled nonlinear equations, e.g., \cite{NgwaB}. A rigorous analysis of such models in their original form is often next to impossible. One can, however, observe that the life cycle of mosquitoes is much faster than that of humans and then, under appropriate assumptions, the mosquitoes should reach equilibrium before the human population undergoes any significant change. Thus, it is plausible to discard the part of the model describing the vector's dynamics while replacing the vector variables in the human part of the model with their equilibrium values, obtaining a reduced model.  Models driven by processes occurring at vastly different rates are called multi-scale. The approach described above is often called the quasi-steady-state assumption (QSSA), e.g., \cite{SegSlem}, and has been successfully used on many occasions to simplify complex multi-scale models. In particular, it has been applied to malaria models in \cite{CapMBS,CapBio,RASS,RashVent,RashKooi} and to other multi-scale epidemiological models in \cite{BL,Bak1,BP,Milaine,Kuehn}.   An interesting outcome of the multi-scale approach to malaria models is the derivation of a new form of the infection force, which is of the Holling 2 form, see \cite{RashVent,RashKooi}; as pointed out in \cite{DiekHee}, such an infection force cannot intrinsically appear in diseases where the infectious and susceptible individuals come from the same population.

The mathematical theory specifying the conditions under which QSSA is valid, that is, that the solutions of the original model can be approximated by the solution of the reduced model when the relevant parameters are small, was developed independently in \cite{TVS} and \cite{Fe}. Though based on different formalisms and using different languages, the theories are equivalent. A survey of applications of the latter to mathematical biology can be found in \cite{Hek}, while the life science applications of the former are the content of \cite{BL}. The problem with both approaches is that they provide approximation valid only on finite time intervals, which is not satisfactory in many cases since they do not ensure that one can approximate the long-term dynamics of the original model by that of the reduced one;  even the global solvability of the reduced model does not imply the same for the original one.  Thus, the reduction of complex epidemiological models in most papers is incomplete. A notable example is \cite{Kuehn} where, though the emphasis is on nonhyperbolic points on the slow manifold, the authors consider the behaviour of solutions as $t\to \infty$ by matching the finite time asymptotics with respect to the small parameter with the known long-term asymptotics of the full system of equations.

The required extension of the theory, ensuring that the approximation is uniform on unbounded time intervals without stability assumptions on the full system, came in \cite{Hop} (see also \cite[Appendix C.18]{Khal});  different, direct proofs were given in \cite{MarCz, BanViet}.

The aim of this paper is to illustrate the applicability of this theory to models of vector-borne diseases where, as noted earlier, the multiscale character is due to the different vital rates of the host and vector populations. For this purpose, we have selected two models. The first one, a minimalistic model of a nonlethal disease (like dengue) with constant vector and host populations, has been introduced and analysed from the multi-scale analysis point of view in \cite{RASS, RashVent}. The authors based their theoretical analysis on the Fenichel theory, while the conclusions on the long-term approximation followed from numerical experiments. A novel aspect of \cite{RashVent} is a higher-order approximation of the slow manifold of the original flow obtained using the manifold's invariance property. We note that an alternative approach to such an approximation using a renormalization group method can be found in \cite{MarCz}, and by the so-called
Chapman--Enskog expansion in \cite{BanRG}; we shall briefly compare them in Appendix \ref{ChE}. Here we complete the analysis of \cite{RASS,RashVent} by showing that the QSS approximation is uniformly valid for all $t\geq 0$ when the basic reproduction number satisfies $\mc R_0<1$ as well as when $\mc R_0>1$.

The second model extends the first one, allowing the host and the vector populations to vary, and introducing disease-induced mortality in the host population.  The former naturally complicates both forces of infection, which are assumed to have the so-called standard form, corresponding to a large host population, see \cite[Table 2.3]{Chitnis}, while the latter provides additional coupling in the system. Nevertheless, also in this case, we can show that the reduced three-dimensional model provides a satisfactory approximation to the original one for all $t\geq 0$ in both regimes: $\mc R_0<1$ and $\mc R_0>1.$

\section{Analysis of some models of vector-borne diseases}

\begin{figure}
\begin{center}
\begin{tikzpicture}[scale=0.65]
\draw (-1,-5)--(1,-5)--(1,-3)--(-1,-3)--(-1,-5);
\draw (-5,-5)--(-3,-5)--(-3,-3)--(-5,-3)--(-5,-5);
%\draw (5,-5)--(3,-5)--(3,-3)--(5,-3)--(5,-5);
\draw [->] (0,-6)--(-6,-6)--(-6,-4)--(-5,-4);
%\draw [->] (-6,-4)--(-5,-4);
\draw [->] (-3,-4)--(-1,-4);
%\draw [->] (1,-4)--(3,-4);
\draw [->] (-4,-3)--(-4,-2);
\draw [->] (0,-3)--(0,-2);
\draw [->] (-4,-5)--(-4,-6);
\draw [->] (0,-5)--(0,-6);
\draw(-4,-4) node {$S_v$};\draw(-3.25,-2.5) node {$\hat\mu_v S_v$};\draw(-2,-3.5) node {$\hat \la_v S$};\draw(0,-4) node {$I_v$};\draw(-2,-5.5) node {$\hat b_v(N_v)$};\draw(0.8,-2.5) node {$\mu_vI_v$};
\draw (-1,-1)--(1,-1)--(1,1)--(-1,1)--(-1,-1);
\draw (-5,-1)--(-3,-1)--(-3,1)--(-5,1)--(-5,-1);
%\draw (5,-1)--(3,-1)--(3,1)--(5,1)--(5,-1);
\draw (-9,-1)--(-7,-1)--(-7,1)--(-9,1)--(-9,-1);
\draw [->] (0,-1.5)--(-10,-1.5)--(-10,0)--(-9,0);
\draw [->] (-7,0)--(-5,0);
%\draw [->] (-10,0)--(-9,0);
\draw [->] (-3,0)--(-1,0);
%\draw [->] (1,0)--(3,0);
\draw [->] (-4,1)--(-4,2);
\draw [->] (-8,1)--(-8,2);\draw [->] (0,1)--(0,2);
\draw [->] (0,-1)--(0,-1.5);
%\draw [->] (4,-1)--(4,-1.5);
\draw [->] (-4,-1)--(-4,-1.5);
\draw [->] (-8,-1)--(-8,-1.5);
%\draw [->] (4,-1)--(4,-1.5);
\draw(-8,2.5) node {$\mu_h S_h$}; \draw(-8,0) node {$S_h$};\draw(-6,0.5) node {$\la_h S_h$};
\draw(-4,0) node {$I_h$};\draw(-4,2.5) node {$(\mu_h +\mu_d)I_h$};\draw(-2,0.5) node {$\gamma_h R_h$};\draw(0,0) node {$R_h$};\draw(0,2.5) node {$\mu_h  R_h$};\draw(-2,-1.25) node {$b_h(N_h)  +\rho_h R_h$};
\draw [dashed, ->] (-1,-3.2)--(-6.95,-0.2);
\draw [dashed, ->] (-5,-0.5)..controls(-6,-2)..(-5,-3.5);
\end{tikzpicture}
\caption{A $S_hI_hR_hS_vI_v$ malaria model}\label{Fig1}
\end{center}
\end{figure}
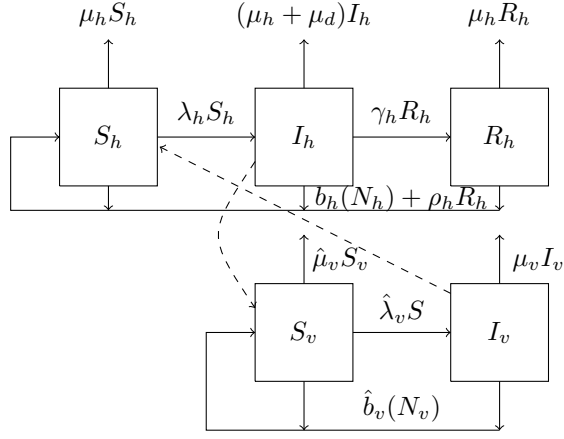
As noted in the Introduction, we consider a minimalistic model of a vector-borne disease. In general, we divide the host population into the following classes: susceptibles $S_h$, infectives $I_h$ and recovered with (waning) immunity $R_h$,  while the vector population is divided into susceptibles $S_v$, and infectives $I_v$. We describe the dynamics of the process, shown in Fig. \ref{Fig1}, by the system
\begin{equation}
\begin{split}
{S}'_h&= b_h( N_h) -\mu_{h} S_h  +\rho_hR_h -  \lambda_{h}S_h,\\
{I}'_h&= \lambda_{h}S_h - (\gamma_h+ \mu_{h}+\mu_{d}) I_h, \\
{R}'_h &= \gamma_h I_h - (\rho_h +\mu_{h})R_h,\\
{S}'_v&= \hat b_v(N_v)-\hat \mu_{v}S_v  - \hat\lambda_{v} S_v,  \\
{I}'_v&= -  \hat\mu_{v}I_v +  \hat\lambda_{v}S_v,
\end{split}\label{C01main-Mal01}
\end{equation}
where $\mbox{}'$ denotes the derivative with respect to time. Here, $N_h$ and $N_v$ denote the total host and vector populations, respectively, and $\hat{\mbox{}}$ is introduced to label `fast' processes.
The infection rates, see \cite{Chitnis,chit2005}, are given by
\begin{equation*}
    \la_h = \vartheta_h(N_h,N_v)\beta_{hv}\frac{I_v}{N_v}, \quad \hat\la_v = \hat \vartheta_v(N_h,N_v)\beta_{vh}\frac{I_h}{N_h},
    \end{equation*}
where $\vartheta_h$ and $\hat \vartheta_v$ are, respectively, the number of bites a human can receive and the number of bites a vector can inflict in a unit of time and $\beta_{hv}$ and $\beta_{vh}$ are respective probabilities of infection. The form of $\vartheta_h$ and $\hat \vartheta_v$ was proposed in \cite{Chitnis} and derived using a Holling type argument in \cite{JBBiomath}. The explicit form of $\la_h$ and $\hat\la_v$ depends on many factors, see \cite{Chitnis}; in this paper, we consider two particular models described in respective subsections.

 The vector's total birth rate is denoted by $\hat b_v$, and its death rate per capita is $\hat\mu_v.$  For the host, the total birth rate is denoted by $b_h$, and its natural death rate per capita is $\mu_h$, while the disease-induced death rate is $\mu_d$. The host's recovery rate (per capita) is $\gamma_h,$ and the immunity loss rate is $\rho_h$.

We assume all parameters are positive, with typical values in Table \ref{Params}, from where
\begin{table}[tbp] %\label{Tablepar}
\caption{Parameter values, \cite[Chapter 3]{chit2005}  }
\label{Params}
\begin{center}
{\footnotesize
\begin{tabular}{|c|c|c|}
\hline
Parameters &  day$^{-1}$ & year$^{-1}$ \\ \hline
%$\Psi_h$ & $7.666\times 10 ^{-5}$ & $2.8 \times 10 ^{-2}$\\ \hline
%$\nu_h$ &  $8.33 \times 10 ^{-2}$& $2.19 \times 10 ^{1}$  \\ \hline
$\gamma_h$ & $1.4\times 10 ^{-3} -17 \times 10 ^{-3}$& $0.511-6.205$ \\ \hline
$\mu_d$ & $0 - 4.1 \times 10 ^{-4}$& $0-1.5 \times 10 ^{-1}$\\ \hline
$\rho_h$ &  $5.5\times 10^{-5}- 1.1 \times 10 ^{-2}$ & $0.02-4.05 $\\ \hline
$\mu_{h}$ & $3.4 \times 10 ^{-5} - 6.8\times 10 ^{-5}$ &  $1.25\times 10^{-2}-2.5 \times 10 ^{-2}$ \\ \hline
%$\mu_{2h}$ & $10 ^{-7}$& $3.65 \times 10 ^{-5} $\\ \hline
%$\sigma_h$ & $18$& $6.570 \times 10 ^{3}$ \\ \hline
%$\sigma_v$&  $0.13-0.47$& $ 4.74\times 10^1-1.72 \times 10 ^{2}$\\ \hline
%$\Psi_v$& $0.4 $&$ 1.46\times 10 ^{2}$  \\ \hline
%$\nu_v $& $0.1$ & $3.65 \times 10 ^{1}$  \\ \hline
$\hat\mu_{v}$&$0.05 - 0.28$ & $1.83\times 10 ^{1} - 10.2 \times 10 ^{1}$ \\ \hline
%$\mu_{2v}$&$ 2.279 \times 10 ^{-4}$& $8. 318 \times 10 ^{-2}$\\  \hline
%\hline
    %  $ \hat\beta_{v}$&      $ 0.009-0.3  $         &      $ 3.41-1.1\times 10^2  $   \\ \hline
     %   $ \beta_{h}$&      $ 0.0013 -0.127 $        &         $0.475-46.36$    \\ \hline
\end{tabular}
}
\end{center}
\end{table}
we see that the time scale of the vector population is much faster than that of the host. Hence, assuming that the vector's population has a tendency to reach its equilibrium, only the vector's data at this equilibrium should have an impact on the host disease dynamics. Then, the vector's dynamics follows the human dynamics in the so-called slave mode, \cite{RASS,RashVent}.

Since the ratio of the host to the vector death rates is  $O(10^{-3})$, we rescale the large coefficients by defining $\mu_v = 10^{-3}\hat\mu_v$, $b_v = 10^{-3}\hat b_v$ and $\lambda_v = 10^{-3}\hat \lambda_v,$ and consider the singularly perturbed version of \eqref{C01main-Mal01},
\begin{equation}
\begin{split}
{N}'_h&= b_h( N_h) -\mu_{h} N_h  -\mu_dI_h,\\
{I}'_h&= \lambda_{h}S_h - (\gamma_h+ \mu_{h}+\mu_{d}) I_h, \\
{R}'_h &= \gamma_h I_h - (\rho_h +\mu_{h})R_h, \\
\e {N}'_v&= b_v(N_v)-\mu_{v}N_v,   \\
\e {I}'_v&=  \lambda_{v}S_v -  \mu_{v}I_v.\end{split}\label{mal1a}
\end{equation}
We observe that for $\e = 10^{-3}$, we recover    \eqref{C01main-Mal01}, see also \cite{Bak1,Milaine}.
\subsection{A simple model of a nonlethal vector born disease}
To illustrate the multi-scale approach to models of vector-borne diseases, we begin with completing the analysis in \cite{RASS,RashVent}, using  Theorem \ref{thm1}.

The authors consider an $S_hI_hS_hS_vI_v$ simplification of  \eqref{C01main-Mal01}, see Fig. \ref{Fig1a}, where it is assumed that the infectives recover without immunity, and the vital dynamics of both host and vector populations are Malthusian with equal death and birth rates, so the populations are constant. The authors also assume that the factors in the infection rates, $\vartheta_h(N_h,N_v)\beta_{hv} =: \beta_h$ and $\hat \vartheta_v(N_h,N_v)\beta_{vh} =:\hat\beta_v$, are constants.  These lead to
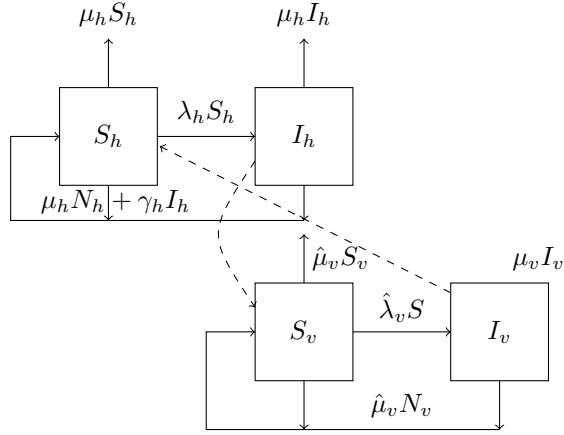
\begin{figure}
\begin{center}
\begin{tikzpicture}[scale=0.65]
\draw (-1,-5)--(1,-5)--(1,-3)--(-1,-3)--(-1,-5);
\draw (-5,-5)--(-3,-5)--(-3,-3)--(-5,-3)--(-5,-5);
%\draw (5,-5)--(3,-5)--(3,-3)--(5,-3)--(5,-5);
\draw [->] (0,-6)--(-6,-6)--(-6,-4)--(-5,-4);
%\draw [->] (-6,-4)--(-5,-4);
\draw [->] (-3,-4)--(-1,-4);
%\draw [->] (1,-4)--(3,-4);
\draw [->] (-4,-3)--(-4,-2);
%\draw [->] (0,-3)--(0,-2);
\draw [->] (-4,-5)--(-4,-6);
\draw [->] (0,-5)--(0,-6);
\draw(-4,-4) node {$S_v$};\draw(-3.25,-2.5) node {$\hat\mu_v S_v$};\draw(-2,-3.5) node {$\hat\la_v S$};\draw(0,-4) node {$I_v$};\draw(-2,-5.5) node {$\hat \mu_v N_v$};\draw(0.8,-2.5) node {$\mu_vI_v$};
%\draw (-1,-1)--(1,-1)--(1,1)--(-1,1)--(-1,-1);
\draw (-5,-1)--(-3,-1)--(-3,1)--(-5,1)--(-5,-1);
%\draw (5,-1)--(3,-1)--(3,1)--(5,1)--(5,-1);
\draw (-9,-1)--(-7,-1)--(-7,1)--(-9,1)--(-9,-1);
\draw [->] (-4,-1.72)--(-10,-1.72)--(-10,0)--(-9,0);
\draw [->] (-7,0)--(-5,0);
%\draw [->] (-10,0)--(-9,0);
%\draw [->] (-3,0)--(-1,0);
%\draw [->] (1,0)--(3,0);
\draw [->] (-4,1)--(-4,2);
\draw [->] (-8,1)--(-8,2);%\draw [->] (0,1)--(0,2);
%\draw [->] (0,-1)--(0,-1.5);
%\draw [->] (4,-1)--(4,-1.5);
\draw [->] (-4,-1)--(-4,-1.72);
\draw [->] (-8,-1)--(-8,-1.72);
%\draw [->] (4,-1)--(4,-1.5);
\draw(-8,2.5) node {$\mu_h S_h$}; \draw(-8,0) node {$S_h$};\draw(-6,0.5) node {$\la_h S_h$};
\draw(-4,0) node {$I_h$};\draw(-4,2.5) node {$\mu_hI_h$};\draw(-7.85,-1.33) node {$\mu_h N_h +\gamma_hI_h$};
\draw [dashed, ->] (-1,-3.2)--(-6.95,-0.2);
\draw [dashed, ->] (-5,-0.5)..controls(-6,-2)..(-5,-3.5);
\end{tikzpicture}
\caption{The $S_hI_hS_hS_vI_v$ model considered in \cite{RASS,RashVent}.}
\end{center}
\label{Fig1a}
\end{figure}
\begin{equation}\label{c01main}
\begin{split}
S'_h&=(\mu_h +\gamma_h)I_h  - \frac{\beta_h}{N_v} S_hI_v,\; \\
I_h'&= \frac{\beta_h}{N_v} S_hI_v - (\mu_h +\gamma_h) I_h,\; \\
S_v'&= \hat\mu_v (N_v-S_v)  -\frac{\hat\beta_v}{N_h}S_vI_h,\; \\
I'_v&= \frac{\hat\beta_v}{N_h}S_vI_h-\hat\mu_v I_v, \; \\
\end{split}
\end{equation}
with  $(S_h(0),I_h(0),S_v(0),I_v(0))= (\mr I_{h}, \mr S_{h},\mr S_{v},\mr I_{v}).$
As discussed above, it is reasonable to transform \eqref{c01main} to the singularly perturbed system
\begin{equation}\label{c01main2}
 \begin{split}
I'_h&= \frac{\beta_h}{N_v} I_v(N_h-I_h) - \alpha_h I_h,\qquad I_h(0)=\mr I_{h},\\
\e I'_v&=\frac{\beta_v}{N_h}I_h(N_v-I_v)-{\mu}_v I_v,\qquad I_v(0)=\mr I_{v},
\end{split}
\end{equation}
where we used the fact that both populations are constant, denoted  $\alpha_h = \mu_h+\gamma_h$ and changed the notation of the `fast' rates, as in \eqref{mal1a}.

The equation for the quasi-steady state, \eqref{deg1},  is $ \frac{\beta_v}{N_h}I_h(N_v-I_v)-{\mu_v} I_v=0,$ hence the slow manifold \eqref{sm} is given by
\begin{equation}
 \mc M =\left\{(I_v,I_h);\;I_v = \frac{\beta_v N_v I_h}{\mu_v N_h +\beta_v I_h}\right\}.\label{mcM}
\end{equation}
The fast dynamics equation, \eqref{auxil'},  takes the form
\begin{equation}
\ti{I}'_{v,\tau}= \frac{{\beta_v}}{N_h}(I_h(N_v-\ti{I}_v)-{\mu_v} \ti I_v),\;\qquad \ti{I}_v(0)=\mr I_v,
\end{equation}
where $I_h$ is treated as a parameter. Note that whenever the derivatives with respect to $t$ and $\tau$ appear side by side, we distinguish them using the notation $\phantom{x}'_{,t}$ and $\phantom{x}'_{,\tau}$, respectively. Since it is a scalar equation, we see that the stability condition \eqref{mug} is satisfied with $\kappa = \mu_v$ and any $\mr I_v\geq 0$ is in the basin of attraction of points of $\mc M$.

Next,  the reduced equation \eqref{deg2} is given by
\begin{equation}
\label{c01main3}
\bar I'_{h,t}=\beta_h\beta_v\frac{(N_h-\bar{I}_h)\bar{I}_h}{ \beta_v\bar{I}_h+\mu_v N_h}-\alpha_h \bar{I}_h,\; \qquad \bar{I}_h(0)=\mr I_{h}.
\end{equation}
Eq. \eqref{c01main3} has the trivial  equilibrium (DFE) $\bar I_h^*=0$ and a unique endemic equilibrium
\begin{equation}
\bar I_h^*  =\frac{N_h(\beta_h\beta_v-\mu_v\alpha_h)}{\beta_v(\beta_h+\alpha_h)},
\label{redeq}
\end{equation}
which is positive if and only if
\begin{equation}
\beta_h\beta_v-\mu_v\alpha_h>0.
\label{EEcond}
\end{equation}
Since it is easy to see that the basic reproduction number is given by
$
\mc R_0 = \frac{\beta_h\beta_v}{\mu_v\alpha_h},$
\eqref{EEcond} expresses the standard condition $\mc R_0>1$ for the emergence of an endemic equilibrium.
Note that the pair $(\bar I^*_h,\bar I_v^*)$, where
\begin{equation}
\bar I_v^* = \frac{\beta_v N_v \bar I^*_h}{\mu_v N_h +\beta_v \bar I^*_h} = \frac{N_v(\beta_h\beta_v-\mu_v\alpha_h)}{\beta_h(\mu_v + \beta_v)}
\label{barIst}
\end{equation}
is also an endemic equilibrium for \eqref{c01main2} (for any $\e>0$) if \eqref{EEcond} is satisfied.

Following \eqref{tiv0}, the zeroth order term of the initial layer can be written as
$$
\ti I_{v,0}(\tau) = \ti I_v(\tau) - \frac{\beta_v N_v \mr I_h}{\mu_v N_h +\beta_v \mr I_h},
$$
where
$$
\ti{I}'_{v,\tau}= \frac{{\beta_v}}{N_h}(\mr I_h(N_v-\ti{I}_v)-{\mu_v} \ti I_v = \frac{\beta_vN_v\mr I_h}{N_h} -\ti I_v\frac{\mu_vN_h +\beta_v\mr I_h}{N_h},\;\qquad \ti{I}_v(0)=\mr I_v.
$$
Hence
\begin{equation}\label{Eq30}
\ti I_{v,0}(\tau) = e^{-\frac{\mu_vN_h +\beta_v\mr I_h}{N_h}\tau}\left(\mr I_v - \frac{\beta_vN_v\mr I_h}{\mu_vN_h +\beta_v\mr I_h}\right).
\end{equation}
Moreover, after some algebra, the term \eqref{CE5} is given by
\begin{equation}
\bar I_{v,1}(t) = -\mu_v\beta_v N_h^2 N_v \frac{I_h (N_h(\beta_h\beta_v-\alpha_h\mu_v)-I_h\beta_v(\alpha_h+\beta_h))}{(\beta_vI_h +\mu_v N_h)^4}.
\label{barIv1}
\end{equation}
Hence, the Chapman-Enskog approximate equation of the bulk part, \eqref{CE5a}, is
\begin{equation}\label{ChEI}
\begin{split}
&\bar I'_{h,\e}  = \beta_h\beta_v\frac{(N_h-\bar{I}_{h,\e})\bar{I}_{h,\e}}{ \beta_v\bar{I}_{h,\e}+\mu_v N_h}-\alpha_h \bar{I}_{h,\e}\\
&\phantom{xx}-\e \mu_v\beta_v\beta_h N_h^2 \frac{\bar I_{h,\e} (N_h(\beta_h\beta_v-\alpha_h\mu_v)-\bar I_{h,\e}\beta_v(\alpha_h+\beta_h))}{(\beta_v\bar I_{h,\e} +\mu_v N_h)^4} (N_h-\bar I_{h,\e}).
\end{split}
\end{equation}
Then, by \eqref{tiu1} and \eqref{tiu1iv}, the  initial layer corrector to $I_h$ is given by
\begin{equation}\label{bIh1}
\begin{split}
\ti I_{h,1}(\tau)  =  -\frac{N_h\beta_h(N_h - \mr I_h)}{N_v(\mu_vN_h +\beta_v \mr I_h)^2}\left(\mu_v\mr I_v N_h+ \beta_v\mr I_h\mr I_v-\beta_v N_v \mr I_h\right)e^{-\frac{\mu_vN_h +\beta_v\mr I_h}{N_h}\tau},
\end{split}
\end{equation}
and hence the corrected initial condition for $\bar I_{h,\e}$, \eqref{baruic}, is
\begin{equation}\label{bIh0}
\bar I_{h,\e}(0) = \mr I_h - \e \frac{N_h\beta_h(N_h - \mr I_h)}{N_v(\mu_vN_h +\beta_v \mr I_h)^2}\left(\mu_v\mr I_v N_h+ \beta_v\mr I_h\mr I_v-\beta_v N_v \mr I_h\right).
\end{equation}
For shall refrain from writing down an explicit formula for $\ti I_1(\tau)$ due to its length.

Summarising,  from Theorem \ref{thm1} we infer
\begin{corollary}\label{coro3.1}
Let $\bar I_{h,\e}$ be the solution to  \eqref{ChEI} with the initial condition \eqref{bIh0} and  $\ti I_{h,1}$ by given by \eqref{bIh1}. Then for the solution $(I_{h,\e},I_{v,\e})$ of \eqref{c01main2} we have
\begin{equation}
I_{h,\e}(t) = \bar I_{h,\e}(t) + \e \ti I_{h,1}\left(\frac{t}{\e}\right) + O(\e^2),
\label{CEest1}
 \end{equation}
 uniformly on $[0,\infty)$ and, with
 \begin{equation}
 \bar I_{v,0}(t) = \frac{\beta_v N_v \bar I_{h,\e}(t)}{\mu_v N_h +\beta_v \bar I_{h,\e}(t)}
 \label{Iv0}
 \end{equation}
 and  $\bar I_{v,1}(t)$ given by \eqref{barIv1} (with $I_h$ replaced by $\bar I_{h,\e}$),
 we have
 \begin{equation}
 I_{v,\e}(t) = \bar I_{v,0}(t) +\e \bar I_{v,1}(t) +O(\e^2),
\label{Ive}
\end{equation}
uniformly on any interval $[t_0,\infty)$, $t_0>0$. The estimate on $[0,\infty)$ can be achieved by adding initial layer terms $\ti I_{v,0}$ (to get $O(\e)$ error) and $\ti I_{v,1}$ (for an $O(\e^2)$ error).
\end{corollary}

\subsection{Numerical simulations}\label{ss22}
Next, we provide a numerical illustration of some of the results derived above. We  use the initial conditions $I_h(0) = 100$,
$I_v(0) = 2000$, $N_h(0) = 10000$ and $N_v(0) = 50000$.  The parameter values are within the ranges presented in Table \ref{Params}, and we used $\beta_v=0.05, \beta_h=1.1$ in the case $\mc R_0<1$, and $\beta_h =4.4,
\beta_v =0.18$  in the case $\mc R_0<1$.
\begin{figure}[ht]
\begin{center}
\includegraphics[scale=0.6]{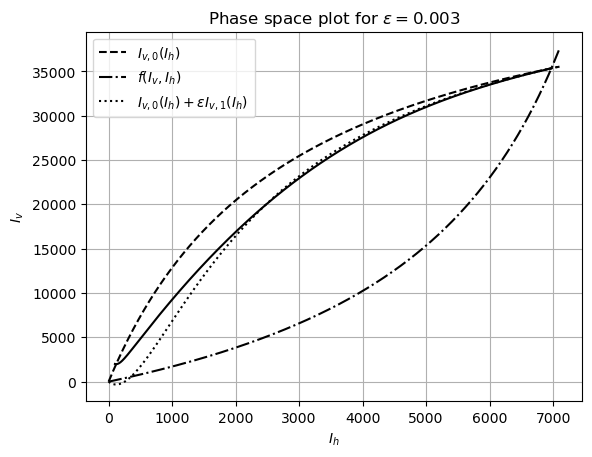}
\end{center}
\caption{Phase-space plot of system \eqref{c01main2} describing the $S_hI_hS_vI_v$ model for $\e =0.005$. The solid line is the trajectory starting at the point $(I_h(0), I_v(0))$. The dashed curve is the nullcline $f (I_v, I_h) = 0$ given in \eqref{c01main2}  and the dotted curve is given by $\bar{I}_{v,0}+\e \bar{I}_{v,1}$  as a function of $I_{h}$. }
\label{fig3}
\end{figure}
\begin{figure}[ht]
\begin{center}
\includegraphics[scale=0.7]{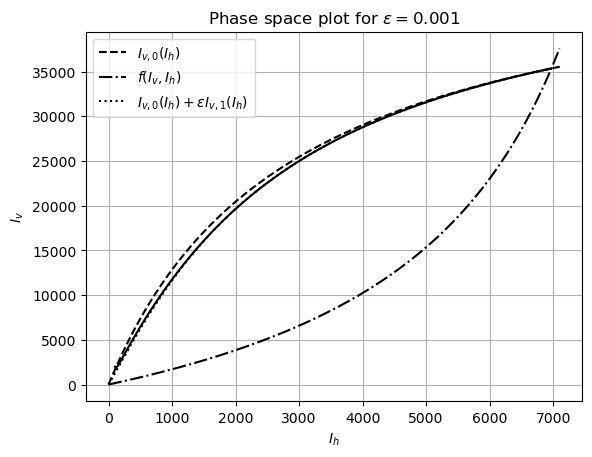}
\end{center}
\caption{An improvement of the approximation shown in Fig. \ref{fig3} for $\epsilon = 0.001.$ The approximating trajectory practically coincides with the trajectory of \eqref{c01main2}.}
\label{I_h_vs_I_v}
\end{figure}

Next, in Figures \ref{Ih_case1} -- \ref{Iv_case3_001}, we provide more detailed simulations. In each figure, we present results for $\mathcal{R}_0>1$ on the left and for $\mathcal{R}_0<1$ on the right.
\begin{figure}[!h]
\begin{center}
\includegraphics[scale=0.4]{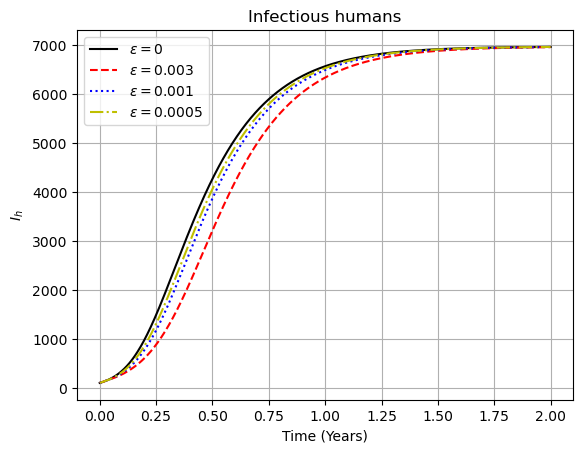}
\includegraphics[scale=0.4]{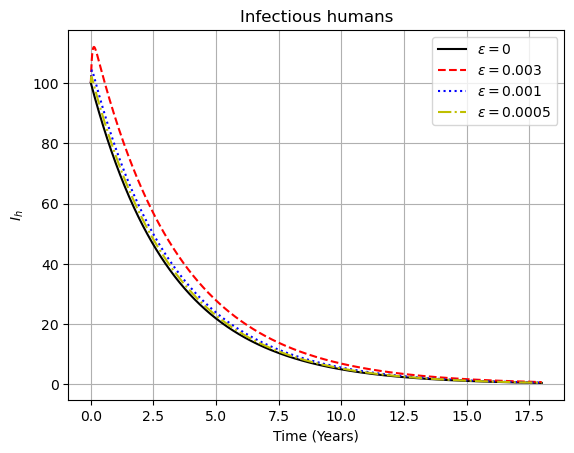}
\end{center}
\caption{Graphs of $I_{h,\e}(t),$ the solutions of \eqref{c01main2} for $\epsilon = 0.005$, $0.001$ and $0.0005$, and the solution $I_h(t)$ of equation \eqref{c01main3} represented by $\epsilon = 0.$  }
\label{Ih_case1}
\end{figure}

\begin{figure}[!h]
\begin{center}
\includegraphics[scale=0.4]{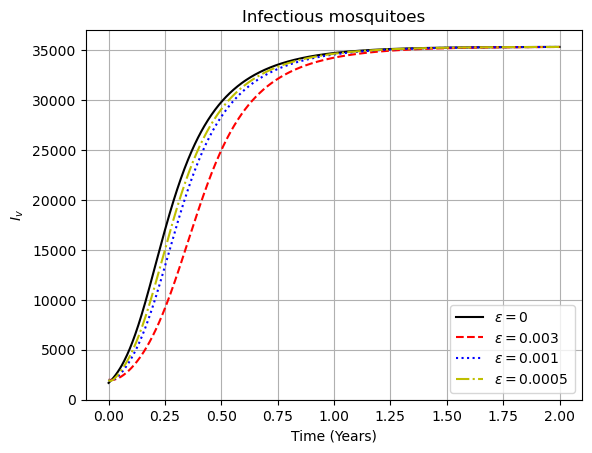}
\includegraphics[scale=0.4]{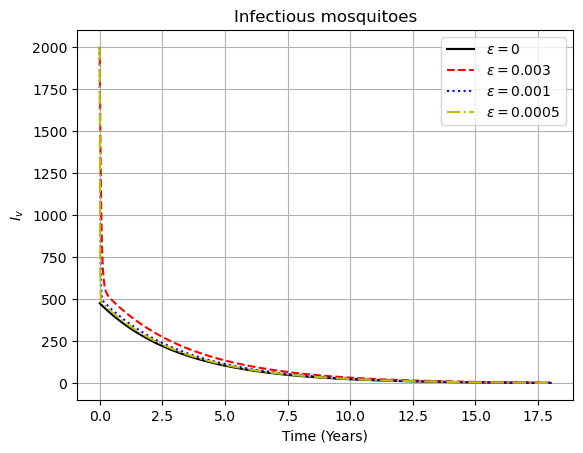}
\end{center}
\caption{Graphs of $I_{v,\epsilon}(t)$, the solutions of \eqref{c01main2} for $\epsilon = 0.005, 0.001\;\text{and}\; 0.0005$, against $\bar{I}_{v,0}(t)$ given by \eqref{Iv0}, represented by $\epsilon = 0.$   We clearly see the initial layer effect caused by the mismatch of initial conditions for $I_{v,\e}$ and $\bar I_{v,0}.$}
\label{Iv_case1}
\end{figure}

%\textbf{Case 2}

%\begin{figure}[!h]
%\begin{center}
%\includegraphics[scale=0.4]%{I_v2_005_Rg1.png}
%\includegraphics[scale=0.4]%{I_v2_005_Rl1.png}
%\end{center}
%\caption{These figures show in red dotted line the graph of the solution $I_{v,\epsilon}(t)$ of \eqref{c01main2} and in black solid line the representation of $\bar{I}_{v,0}(t)+\tilde{I}_{v,0}(\tau)$ given in equations \eqref{Iv0} and \eqref{Eq30} where $\epsilon=0.005$.}
%\label{Iv_case2_005}
%\end{figure}

\begin{figure}[!h]
\begin{center}
\includegraphics[scale=0.4]{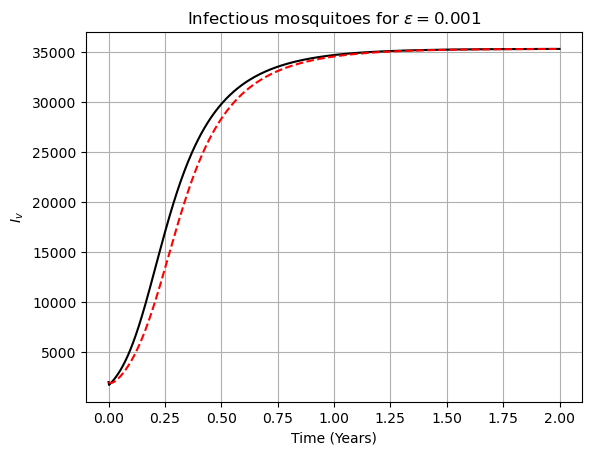}
\includegraphics[scale=0.4]{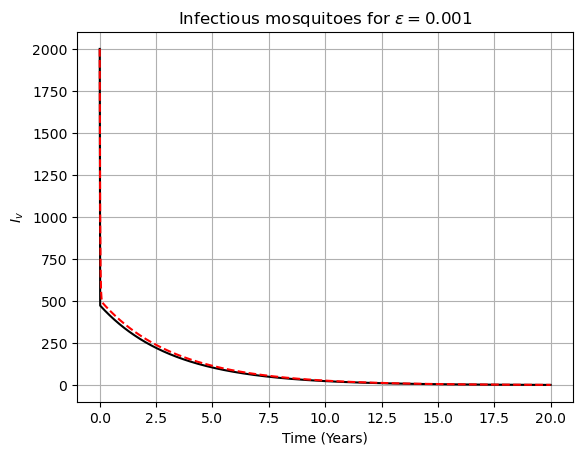}
\end{center}
\caption{The above figures show the improvement in the approximation presented in Fig. \ref{Iv_case1}, achieved by adding the initial layer correction. The red dotted line is the graph of the solution $I_{v,\epsilon}(t)$ of \eqref{c01main2} and the black solid line is the graph of $\bar{I}_{v,0}(t)+\tilde{I}_{v,0}(\tau),$ given by \eqref{Iv0} and \eqref{Eq30} in the case  $\epsilon=0.001$.}
\label{Iv_case2_001}
\end{figure}

%\begin{figure}[!h]
%\begin{center}
%\includegraphics[scale=0.4]%{I_v2_0005_Rg1.png}
%\includegraphics[scale=0.4]{I_v2_0005_Rl1.png}
%\end{center}
%\caption{These figures show in red dotted line the graphs of the solutions $I_{v,\epsilon}(t)$ of \eqref{c01main2} and in black solid line the representation of $\bar{I}_{v,0}(t)+\tilde{I}_{v,0}(\tau)$ given in equations \eqref{Iv0} and \eqref{Eq30} where $\epsilon=0.0005$: On the left we have the case $\mathcal{R}_0>1$ and at the right side the case $\mathcal{R}_0<1$ (in this case, $\beta_v=0.05$, $\beta_h=1.1$).}
%\label{Iv_case2_0005}
%\end{figure}

\newpage
%\textbf{Case 3}

\begin{figure}[!h]
\begin{center}
\includegraphics[scale=0.4]{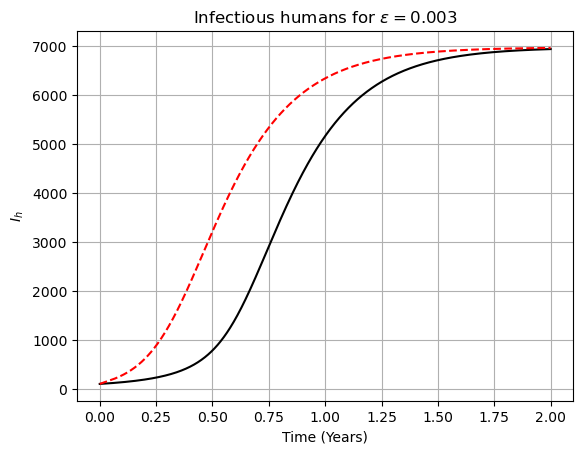}
\includegraphics[scale=0.4]{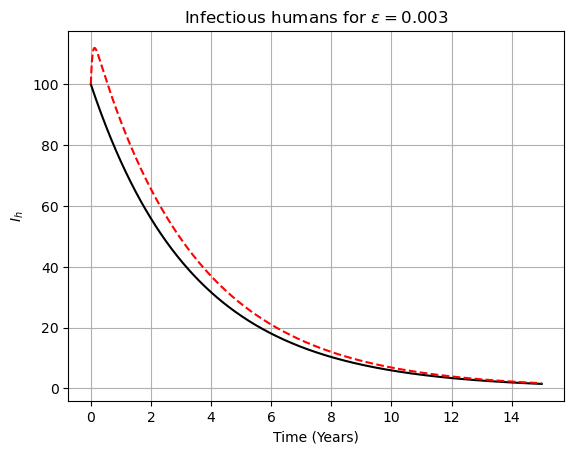}
\end{center}
\caption{These figures show in red dotted line the graphs of the solutions $I_{h,\epsilon}(t)$ of \eqref{c01main2} and in solid black line the solution of the equation \eqref{ChEI} where $\epsilon=0.003$.}
\label{Ih_case3_005}
\end{figure}

\begin{figure}[!h]
\begin{center}
\includegraphics[scale=0.4]{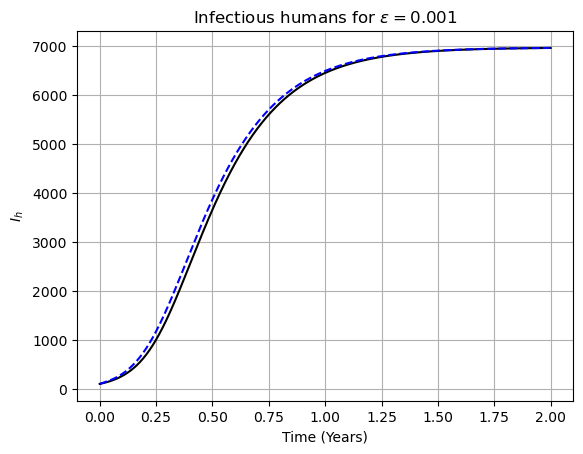}
\includegraphics[scale=0.4]{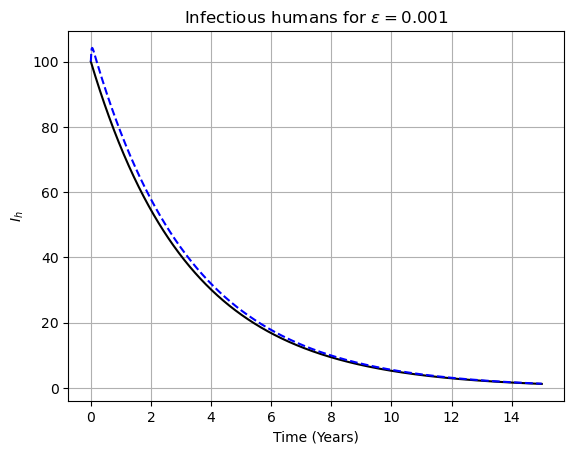}
\end{center}
\caption{Improvement of the approximation shown in Fig. \ref{Ih_case3_005} when $\epsilon =0.001$. The blue dotted line is the graph of the solutions $I_{h,\e}(t)$ of \eqref{c01main2} and the black solid line is the solution of the equation \eqref{ChEI}. }
\label{Ih_case3_001}
\end{figure}

%\begin{figure}[!h]
%\begin{center}
%\includegraphics[scale=0.4]{I_h3_0005_Rg1.png}
%\includegraphics[scale=0.4]{I_h3_0005_Rl1.png}
%\end{center}
%\caption{These figures show in yellow dotted line the graphs of the solutions $I_{h,\e}(t)$ of \eqref{c01main2} and in black solid line the solution of the equation \eqref{ChEI}, where $\epsilon=0.0005$: On the left, we have $\mathcal{R}_0>1$ and at the right,  $\mathcal{R}_0<1$ (in this case, $\beta_v=0.05$, $\beta_h=1.1$).}
%\label{Ih_case3_0005}
%\end{figure}

%\begin{figure}[!h]
%\begin{center}
%\includegraphics[scale=0.4]%{I_v3_003_Rg1.png}%
%\includegraphics[scale=0.4]%{I_v3_003_Rl1.png}
%\end{center}
%\caption{These figures show in red dotted line the graphs of the solutions $I_{v,\epsilon}(t)$ of \eqref{c01main2} and in black solid line the representation of $\tilde{I}_{v,0}(\tau)+\e\bar{I}_{v,1}(t)$ given in corollary \ref{coro3.1} where $\epsilon=0.005$.}
%\label{Iv_case3_005}
%\end{figure}

\begin{figure}[!h]
\begin{center}
\includegraphics[scale=0.35]{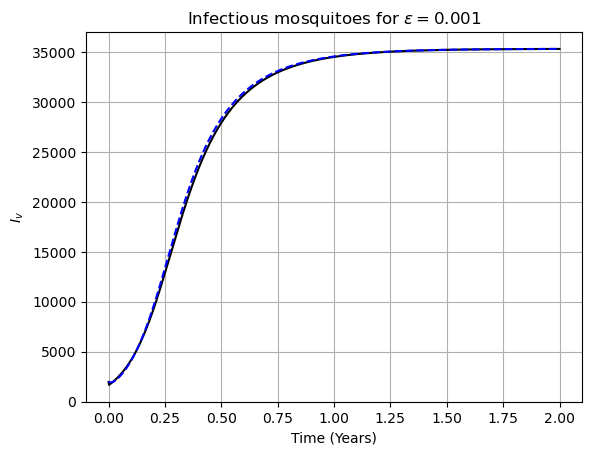}
\includegraphics[scale=0.35]{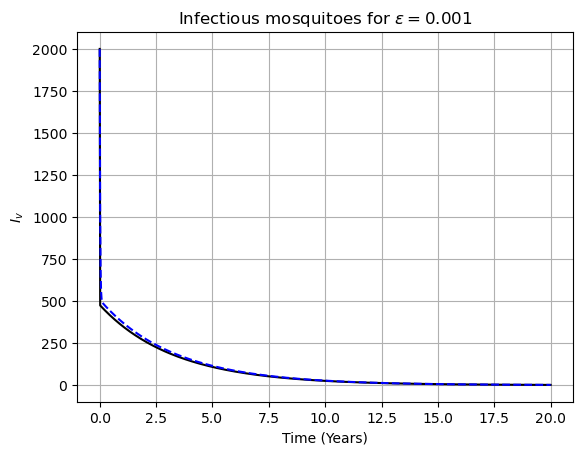}
\end{center}
\caption{These figures show in blue dotted line the graphs of the solutions $I_{v,\epsilon}(t)$ of \eqref{c01main2} and in black solid line the graphs of $\tilde{I}_{v,0}(\tau)+\bar{I}_{v,0}(t)+\e\bar{I}_{v,1}(t)$ given in corollary \ref{coro3.1}, where $\epsilon=0.001$.}
\label{Iv_case3_001}
\end{figure}

%\begin{figure}[!h]
%\begin{center}
%\includegraphics[scale=0.35]%{I_v3_0005_Rg1.png}
%\includegraphics[scale=0.35]%{I_v3_0005_Rl1.png}
%\end{center}
%\caption{These figures show in yellow dotted line the graphs of the solutions $I_{v,\epsilon}(t)$ of \eqref{c01main2} and in black solid line the representation of $\tilde{I}_{v,0}(\tau)+\bar{I}_{v,1}(t)$ given in corollary \ref{coro3.1} where $\epsilon=0.0005$: On the left, we have the case $\mathcal{R}_0>1$ and on the right, $\mathcal{R}_0<1$ (in this case, $\beta_v=0.05$, $\beta_h=1.1$).}
%\label{Iv_case3_0005}
%\end{figure}

%\begin{figure}[!h]
%\begin{center}
%\includegraphics[scale=0.35]{case_1_and_3_003.png}
%\includegraphics[scale=0.35]{case_1_and_3_001.png}
%\end{center}
%\caption{These figures show in yellow dashed line the graphs of the solutions the equation \eqref{ChEI}, in the solid red line the solution of the equation \eqref{c01main2} and in dotted blue line solid the solution of the equation  \eqref{c01main3}.}
%\label{case_1_3}
%\end{figure}

\subsection{The full model \eqref{mal1}}\label{ss2.3}

Here, we consider \eqref{mal1} with non-constant populations and disease-induced death. In the case of variable populations, the dependence of the biting rates on the sizes of respective populations becomes important. Hence, we shall select the explicit form of them, corresponding, according to \cite[Table 2.3]{Chitnis},  to a small, relatively to humans, vector population.  Then
\begin{equation*}
\begin{split}
    \la_h &= b_h(N_h,N_v)\beta_{hv}\frac{I_v}{N_v} = \sigma_v\beta_{hv} \frac{N_v}{N_h}\frac{I_v}{N_v}= :\beta_h\frac{I_v}{N_h}, \\ \hat\la_v &= \hat b_v(N_h,N_v)\beta_{vh}\frac{I_h}{N_h} =\hat\sigma_v \beta_{vh}\frac{I_h}{N_h}=: \hat\beta_v\frac{I_h}{N_h},
    \end{split}
    \end{equation*}
where $\sigma_v,\hat\sigma_v$ are the mosquitoes' biting rates. Hence,  \eqref{mal1a} takes the form \begin{equation}
\begin{split}
{N}'_h&= b_h( N_h) -\mu_{h} N_h  -\mu_dI_h,\\
{I}'_h&= \beta_h\frac{I_vS_h}{N_h} - (\gamma_h+ \mu_{h}+\mu_{d}) I_h, \\
{R}'_h &= \gamma_h I_h - (\rho_h +\mu_{h})R_h, \\
\e {N}'_v&= b_v(N_v)-\mu_{v}N_v,   \\
\e {I}'_v&=  \beta_v\frac{S_vI_h}{N_h} -  \mu_{v}I_v,\end{split}\label{mal1}
\end{equation}
so we effectively deal with a system of four equations, with the equation for the total vector population feeding into the system but not being affected by it.

We assume that the disease-free vector population equation
\begin{equation}
N_v'= b_v(N_v)-\mu_v N_v
\label{vpe}
\end{equation}
has a unique globally attracting positive hyperbolic equilibrium $N_v^*$ (or the populations is constant, that is, $N_v(t) \equiv N_v^*$ for all time).

Next, we  assume that the disease-free human population equation
\begin{equation}
N_h'= b_h(N_h)-\mu_h N_h
\label{hpe}
\end{equation}
has a unique positive hyperbolic equilibrium $N_h^*,$ which is globally stable on $\mbb R_+$, that is, in particular
\begin{equation}
    b_h'(N^*_h) -\mu_h<0.
    \label{bnh}
\end{equation}
As in the previous section,  the quasi-steady state is given as the solution of
\begin{equation}
\begin{split}
0&= b_v(N_v)-\mu_{v}N_v,  \\
0&= -  \mu_{v}I_v +  \beta_{v}\frac{I_h(N_v-I_v)}{N_h},\end{split}\label{mal2}
\end{equation}
from which we get the slow manifold $\mc M$ determined by the equation
\begin{equation}
\bar I_v = N_v^*\frac{\beta_{v}I_h}{\mu_v N_h +\beta_{v} I_h},
\label{qssm}
\end{equation}
where $N_v^*$ is also a unique equilibrium of the fast vector population equation \begin{equation}
\ti N_{v,\tau} = b_v(\ti N_v)-\mu_{v}\ti N_v.\label{Nv}
\end{equation}
We observe that, as in  \eqref{mcM}, $\bar I_v$ is of the Holling 2 form, though this time also $N_h$ is variable.

The Jacobi matrix of the fast part is given by
$$
\left(\begin{array}{cc} b_v'(N_v)-\mu_v&0\\
\beta_{v}\frac{I_h}{N_h}& -\mu_v-\beta_{v}\frac{I_h}{N_h}\end{array}\right)$$
and, since $N_v^*$ is an attractive equilibrium of \eqref{Nv}, see \eqref{vpe},  \eqref{auxil'} is satisfied. Note that in the case of constant $N_v(t) = N_v^*$, \eqref{Nv} is irrelevant, and the mosquito part of the system reduces to a single equation for $I_v,$ as in \eqref{c01main2}.

The reduced system \eqref{deg2} is given by \begin{equation}
\begin{split}
{\bar N}'_h&= b_h(\bar N_h) -\mu_{h} \bar N_h -\mu_d \bar I_h,\\
{\bar I}'_h&= \beta_{h}\beta_{v}N_v^*\frac{\bar N_h-\bar I_h-\bar R_h}{\bar N_h}\frac{\bar I_h}{\mu_v \bar N_h +\beta_{v} \bar I_h} - (\gamma_h+ \mu_{h}+\mu_{d}) \bar I_h, \\
{\bar R}'_h &= \gamma_h \bar I_h - (\rho_h +\mu_{h})\bar R_h. \end{split}\label{mal3}
\end{equation}
The standard analysis gives a  global unique solvability of \eqref{mal3} in the  admissible state space
\begin{equation}
\bar I_h\geq 0,\quad  \bar R_h\geq 0,\quad \bar I_h+\bar R_h\leq \bar N_h.
\label{adpar}
\end{equation}
It is also easy to see that the basic reproduction number for \eqref{mal3} is given by
\begin{equation}
\mc R_0 := \frac{\beta_h\beta_v N_v^*}{N^*_h\mu_v(\gamma_h+\mu_h+\mu_d)}.
\label{R0}
\end{equation}
The main of this section is to show that \eqref{mal1} satisfies the assumptions of Theorem \ref{thm1}.  More precisely, we shall show that assumption \eqref{A5}, that is, the exponential stability of the relevant equilibria of \eqref{mal3},  is satisfied at the disease-free equilibrium if $\mc R_0<1$ and at the endemic equilibrium if $\mc R_0>1$ (at least for low disease-induced death rates).  We shall do this in a series of lemmas.

\begin{lemma} $\bar I_h$ and $\bar R_h$ are bounded (irrespective of $\bar N_h$).\end{lemma}
\begin{proof} We have
$$
0\leq {\bar I}'_h\leq  \beta_{h}N_v^* - (\gamma_h+ \mu_{h}+\mu_{d}) \bar I_h,
$$
so that
$$
0\leq \bar I_h(t) \leq e^{- (\gamma_h+ \mu_{h}+\mu_{d})t}I_0 + \frac{\beta_{h}N_v^*}{(\gamma_h+ \mu_{h}+\mu_{d})}\left(1-e^{- (\gamma_h+ \mu_{h}+\mu_{d})t}\right)
$$
and we get $0\leq \limsup_{t\to\infty}\bar I_h(t) \leq \frac{\beta_{h}N_v^*}{\gamma_h+ \mu_{h}+\mu_{d}}$. \end{proof}
Since no mosquito variables will be used hereafter, we drop the overline symbol $\ov{\phantom{x}}$ and the subscript $h$ to simplify the notation in the sequel.
Let $ N_h^*>0$  be the globally stable equilibrium to \eqref{hpe}.
Then the disease-free equilibrium for \eqref{mal3} is given by
$$
\ov{DFE} = ( N_h^*, 0,0),
$$
(and coincides with the host part of the $DFE$ of the entire system \eqref{C01main-Mal01}).
The Jacobi matrix of \eqref{mal3} at $\ov{DFE}$ is given by
$$
\mc J_{\ov{DFE}} = \left(\begin{array}{ccc} b_h'( N^*_h) - \mu_h&-\mu_d&0\\
0&\frac{\beta_h\beta_v N_v^*}{\mu_v  N^*_h}-(\gamma_h +\mu)&0\\
0&\gamma_h&-(\gamma_h+\mu_h)\end{array}\right),
$$
where we denoted  $\mu =\mu_h+\mu_d$. Thus $\ov{DFE}$ is locally exponentially stable if \eqref{bnh} is satisfied and
\begin{equation}
 N^*_h > \frac{\beta_h\beta_v N_v^*}{\mu_v(\gamma_h+\mu)}\;\Leftrightarrow\; \mc R_0  <1.
\label{ls}\end{equation}
We shall prove that it is also globally stable.  First, we observe that \eqref{mal3} has the following isoclines:
\begin{align}
N'&= 0\;\; \textrm{if}\;\; I = F(N) := \frac{1}{\mu_d}b_h(N) -\frac{\mu_{h}}{\mu_d} N,\nn\\
{I}'&=0 \;\; \textrm{if}\;\;I = G(N,R):= \frac{\beta_{h}\beta_{v}N_v^*N-\mu_v(\gamma_h+\mu) N^2-\beta_{h}\beta_{v}N_v^*R}{\beta_{h}\beta_{v}N_v^*+ \beta_v (\gamma_h+\mu)N}\nn\\& \;\;\;\;\;\;\;\;\textrm{or}\; I =0, \nn\\
{R}'&= 0 \;\; \textrm{if}\;\; I = H(R) := \frac{\rho_h +\mu_{h}}{\gamma_h}R, \label{mal3'}
\end{align}
which are presented in Fig. \ref{Fig 1n}.
  \begin{figure}
\begin{center}
\begin{tikzpicture}[scale=1]
\draw [thick,->](-6,0)--(3.5,0);
\draw [dashed,-] (-4,0)--(-7,3.75);
\draw [dashed,-] (3.5,0)--(0.5,3.75);
\draw [dashed,-] (-7,3.75)--(0.5,3.75);
\draw [thick,->](-4,-1)--(-4,5);
\draw [thick,->] (-4,0)--(-7,-3);
\draw (-7,-2.5) node {\tiny{$R$}};
\draw (-3,3.5) node {\tiny{$I = H(R)$}};
%\draw (5,-0.5) node {x};
%\draw (0.5,5) node {y};
%\draw [-](-4,4)--(0,0);
\draw [-](-4,0)--(2,4.9);
\draw [-](-4,0)--(-3.1,-3);
\draw  [thick,dotted,-] (-4,0) .. controls  (-1.5,4) .. (2,0);
\draw  [thick,dotted,-] (-7,-3) .. controls  (-4.5,1) .. (-1,-3);
\draw [thick,dotted,-] (2,0)--(-1,-3);
\draw [thick,dotted,-] (-1.5,3)--(-4.5,0);
\draw  [dash pattern=on 1pt off 2pt on 3pt off 2pt] (-4,0) .. controls  (-3.5,-1.5) .. (0,0);
\draw [thick, dashed,-] (-4,0).. controls (-2.2,0.7)..(0,0);
\draw [thick, dashed,-] (-4,0)..controls (-2.27,3.31)..(2,0);
\draw [dash pattern=on 1pt off 2pt on 3pt off 2pt] (-4,0) to [out=45,in = 135] (0,0);
\draw [dash pattern=on 1pt off 2pt on 3pt off 2pt] (-2,0.82) -- (-3.2,-1.12);
\draw (3.5,-0.5) node {\tiny{$N$}};
\draw (-4.5,5) node {\tiny{$I$}};
\draw (1.5,5) node {\tiny{$I=N+R$}};
\draw (1,2) node {\tiny{$I=F(N)$}};
\draw (-2,1) node {\tiny{$I=G(N,R)$}};
\draw (0.55,-0.3) node {\tiny{$\frac{\beta_v\beta_hN_v^*}{\mu_v(\gamma_h+\mu)}$}};
\draw [fill]  (0,0) circle [radius=0.05];
\draw (2,-0.5) node {\tiny{$ N_h^*$}};
\draw [fill]  (2,0) circle [radius=0.05];
%\draw [dotted] (-1.05,0)--(-1.05,2.9);
%\draw [dotted] (-2,0)--(-2,2);
%\draw (-2,-0.5) node {$\tiny{\frac{\beta_v\beta_hN_v^*}{(\mu_v+\beta_v)(\gamma_h+\mu)}}$};
%\draw [fill]  (-2,0) circle [radius=0.05];
%\draw [fill]  (-1.05,0) circle [radius=0.05];
%\draw (-1.15,0.2) node {$\tiny{ N^*}$};
\draw [->] (0,2.2)--(0,1.95);
\draw [->] (1,1.12)--(1,0.87);
\draw [->] (0.5,1.665)--(0.5,1.415);
\draw [->] (1.5,0.55)--(1.5,0.3);
\draw [->] (-0.5,2.62)--(-0.5,2.37);
\draw [->] (-0.5,0.4)--(-0.25,0.4);
\draw [->] (-1,0.65)--(-0.75,0.65);
\draw [->] (-1.5,0.795)--(-1.25,0.795);
\draw [->] (-3.5,0.4)--(-3.25,0.4);
\draw (2.5,3) node {\tiny{\begin{tabular}{c}{$ N'<0$}\\{$I'<0$}\\{$R'>0$}\end{tabular}}};
\draw (2.5,-1.5) node {\tiny{\begin{tabular}{c}{$ N'<0$}\\{$I'<0$}\\{$R'<0$}\end{tabular}}};
\draw (-0.6,1.4) node {\tiny{\begin{tabular}{c}{$ N'>0$}\\{$I'<0$}\\{$R'<0$}\end{tabular}}};
\draw (-2.1,-0.15) node {\tiny{\begin{tabular}{c}{$ N'>0$}\\{$I'>0$}\\{$R'<0$}\end{tabular}}};
\end{tikzpicture}
\end{center}
\caption{Isoclines of \eqref{mal3'} with \eqref{ls} satisfied. The dotted tent-like surface corresponds to $I=F(N)$, the dashed plane corresponds to $I=H(R)$ and the dash-dotted surface corresponds to $I=G(N,R)$. The thick dashed curves show the intersections of $I=H(R)$ and $I=F(N)$ (top), and $I=H(R)$ and $I=G(N,R).$ The admissible part \eqref{adpar} is under the plane $I=N+R$.}
 \label{Fig 1n}
\end{figure}
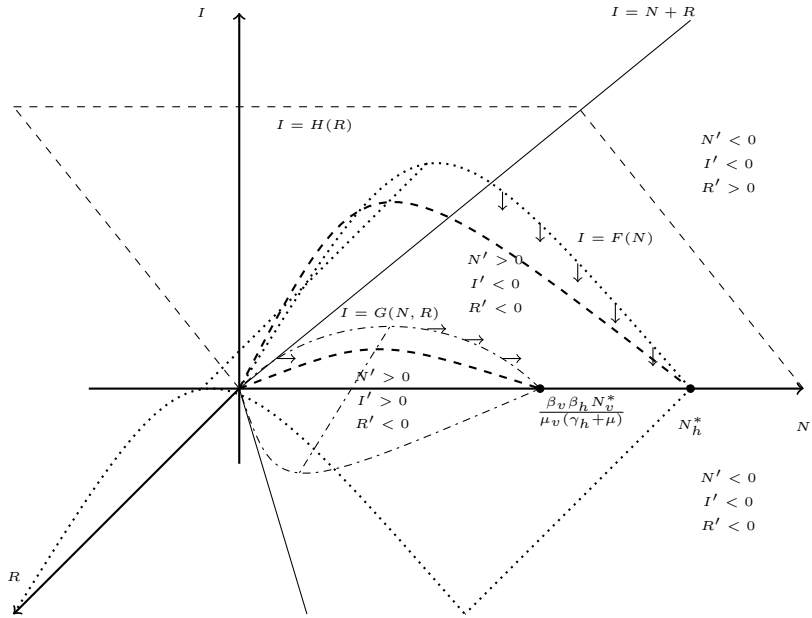Next, we make several obvious observations.
\begin{description}
\item \textit{Observation 1.} For $R\geq 0$,
$$G(N, R)\leq G_0(N):=  N-\frac{\mu_v(\gamma_h+\mu)}{\beta_{h}\beta_{v}N_v^*}  N^2.$$
Thus,  $I= F(N)$ and $I= G(N,R)$ intersect  only if there is an intersection of the curves $I= F(N)$ and $I= G_0(N)$ in the $(I,N)$ plane.
\item \textit{Observation 2.}  $N=0$ is a solution of $F(N) = G_0(N)$. Further,
\begin{align*}
G'_{,N}(0,0) &= G'_{0,N}(0)=1,\\
G_0(N)=0 &\;\;\textsf{if}\;\; N=0\;\;\textsf{or}\;\; N^*:= \frac{\beta_h\beta_v N_v^*}{\mu_v(\gamma_h+\mu)}.
\end{align*}
\end{description}
Next, we have to make another assumption concerning demography. Namely, in addition to \eqref{ls}, we assume
\begin{equation}
F(N)>G_0(N)\quad \textrm{for} \quad 0<N< \frac{\beta_h\beta_v N_v^*}{\mu_v(\gamma_h+\mu)}.
\label{Bass}
\end{equation}
\begin{remark}For the logistic birth rate,  $b_h(N) = rN\left(1-\frac{N}{K}\right), r>0,K>0,$  assumption \eqref{Bass} is satisfied if (and only if)
\begin{equation}
\theta :=\frac{r-\mu_h}{\mu_d} \in (0,1).
\label{log1}
\end{equation} Indeed, in this case,  the equilibrium is given by
$$
 N^*_h = \frac{K(r-\mu_h)}{r},
$$
so that $r-\mu_h>0$ by \eqref{ls}, hence $\theta$ must be positive. The nonzero intersection of $I=F(N)$ and $I=G_0(N)$ can be obtained by solving
$$
\theta-1 = N\left(\frac{\theta}{ N^*_h}-\frac{1}{N^*}\right),
$$
that is, it is given by
$$
N_0 = \frac{(\theta - 1) N^*_hN^*}{\theta N^*- N^*_h}, \quad \theta \neq  \frac{N^*_h}{N^*}.
$$
We see that $\theta =1$ gives additional solution $N_0=0$, while  $\theta =  N^*_h/N^*$ yields no solution.  Then $0<N_0< N^*_h$ if and only if
$$
0<
\frac{(\theta - 1)N^*}{\theta N^*- N^*_h} = \frac{(\theta - 1)N^*}{(\theta -1)N^* +N^*- N^*_h}
<1.
$$
If $0<\theta <1$, then $0<N_0< N^*_h,$ as, by \eqref{ls}, we have $N_h^*>N^*$ and hence
$$
\frac{(\theta - 1)N^*}{(\theta -1)N^* +N^*- N^*_h}=\frac{(1-\theta)N^*}{(1-\theta)N^* +N^*_h-N^*}<1. $$ If $\theta>1$,  $N_0 < 0$ if $\theta< N^*_h/N^*$ and $N_0> N^*_h$ if $\theta >  N^*_h/N^*$.\end{remark}
\begin{lemma} The admissible part, see \eqref{adpar}, of the `dotted' tunnel,
$$\Omega_0 = \{(N,I,R);\; I\geq 0, R\geq 0, I+R\leq N, I\leq F(N)\}$$
is invariant under the flow of \eqref{mal3}.\label{lem42} \end{lemma} \begin{proof}Indeed, on the surface $N'= 0$, that is, $\{(N,I,R);\;I=F(N), I,R\geq 0\},$ we have $I'<0$ unless $N=R=I=0$ (and thus the field points inward), and the origin is repelling. \end{proof}
\begin{lemma} No trajectory can pass from the admissible part of $R'<0$ to $R'> 0$, that is, across $\{(N,I,R);\;I= H(R)\},$ outside the 'dash-dotted' region $I> G(N,R)$. \label{lem43}\end{lemma}\begin{proof}
Indeed, the normal on $I=H(R)$ pointing towards the $R'>0$ is given by $\left(0,1,-\frac{\rho_h+\mu_h}{\gamma_h}\right)$ and the trajectory's direction at the boundary is $(N',I',0)$ with $I'<0$. \end{proof}
\begin{lemma} Any trajectory originating in $\Omega_0\setminus \{0,0,0\}$ converges to $( N^*_h,0,0)$. \label{lem44}\end{lemma} \begin{proof} First, we observe that $\Omega_0$ is invariant by Lemma \ref{lem42}, and bounded thanks to $N\geq I+R$. Thus the trajectories $t\mapsto ( N(t), I(t), R(t))$ are bounded. Hence, $$t\mapsto (F( N(t),  I(t)), G(N(t), I(t),  R(t)), H(I(t),  R(t)))$$
is bounded (note that $0\leq \frac{ N- I- R}{N} \leq 1$) and thus
$
t\mapsto ( N'(t),  I'(t),  R'(t))$
is uniformly continuous. For any trajectory staying in $\Omega_0$ we have $ N'>0,$ so there exists $$\lim_{t\to \infty}  N(t) =: \bar N.$$ But then, due to the uniform continuity of $ N'$, see \cite[Corollary 2.1]{Mar},  $$\lim_{t\to \infty} N'(t)=0.$$ Then the first equation of \eqref{mal3}  implies $$\lim_{t\to \infty} I(t) = :\bar I$$
 and, using again \cite[Corollary 2.1]{Mar}, $$ \lim_{t\to \infty} I'(t)=0.$$
In the same way, but using the second equation of \eqref{mal3}, we get
$$\lim_{t\to \infty}  R(t) =: \bar R\quad\textrm{ and}\quad  \lim_{t\to \infty} R'(t)=0.$$
Thus, $(\bar N,\bar I,\bar R)$ is an equilibrium. Since there are no other attracting equilibria,
$$
(\bar N, \bar I, \bar R) = (\bar N^*_h,0,0).
$$
\end{proof}
\begin{lemma} Any trajectory originating in the region $$\Omega_1= \{(N,I,R);\; I\geq 0, R\geq 0, I+R\leq N, I> F(N)\}$$
converges to $( N^*_h,0,0)$. \end{lemma} \begin{proof} Consider first a trajectory starting in $\Omega_{1a}=\{(N,I,R);\; I\geq 0, R\geq 0, I+R\leq N, I> F(N), I\geq H(R)\}$. The solution is bounded if it stays in $\Omega_{1a}$. Hence, it must have a limit. The only possible limit is at DFE, but this is impossible as $0\leq  R(t)$ is increasing in $\Omega_{1a}$, while  $\lim_{t\to \infty}  R(t) = 0.$ Thus the trajectory must leave $\Omega_{1a}$. If it leaves through $I=F(N)$, we have the situation described in Lemma \ref{lem44}. Otherwise, the trajectory enters
$$\Omega_{1b}= \{(N,I,R);\; I\geq 0, R\geq 0, I+R\leq N, I> F(N), I<H(R)\}$$
and, by Lemma \ref{lem43}, it can cross back to $\Omega_{1a}$. Hence, it can either stay in $\Omega_{1a}$ or again cross $I=F(N)$ to enter   $\Omega_0$. The only new case is the former,  but then each $N(t), I(t)$ and $R(t)$ remain decreasing and bounded from below and hence convergent to the DFE (which does not lead to contradiction if $N(0) (=N_h(0)) >  \bar N^*$.)\end{proof}
\subsection{The asymptotics in the endemic regime}
We consider the approximation in the endemic regime for low disease-induced death rates, that is, for the following version of \eqref{mal1}
\begin{equation}
\begin{split}
{N}'_h&= b_h( N_h) -\mu_{h} N_h  -\e\mu'_dI_h,\\
{I}'_h&= \beta_{h}\frac{I_vS_h}{N_h} - (\gamma_h+ \mu_{h}+\e\mu'_{d}) I_h, \\
{R}'_h &= \gamma_h I_h - (\rho_h +\mu_{h})R_h, \\
\e {N}'_v&= b_v(N_v)-\mu_{v}N_v,   \\
\e {I}'_v&=  \beta_{v}\frac{I_hS_v}{N_h} -  \mu_{v}I_v,\end{split}\label{mal1e}
\end{equation}
where, similarly as in \eqref{mal1}, we normalized $\mu'_d = 10^3\mu_d $ so that $\e\mu_d' =\mu_d$ for $\e = 10^{-3}.$
As the vector part of the system \eqref{mal1e} is the same as in \eqref{mal1}, the reduced model is given by
\begin{subequations}
\begin{equation}
\hspace{-1cm} { N}'= b_h( N) -\mu_{h}  N,\label{m1}
\end{equation}
and
\begin{equation}
\begin{split}
{ I}'&= \beta_{v}\beta_{h}N_v^*\frac{ N- I- R}{ N}\frac{ I}{\mu_v  N +\beta_{v}  I} - (\gamma_h+ \mu_{h})  I =:f(I,R), \\
{ R}' &= \gamma_h  I - (\rho_h +\mu_{h}) R =:g(I,R),
\end{split}
\label{m2}
\end{equation}
\label{mal3a}
\end{subequations}
where, as before, we dropped the index $h$ and the bars $\ov{\phantom{x}}$ in the reduced equation.
To get endemic equilibria of \eqref{mal3}, we first observe that the population's endemic equilibrium $N^*$, solving \eqref{m1},  is the same as in the $\ov{DFE}$ and then we immediately get
\begin{equation}
R^* = \frac{\gamma_h}{\rho_h+\mu_h} I^*.
\label{rs}
\end{equation}
Then, the nontrivial $I^*$ is given by
\begin{equation}
I^* = N^*(\mu_h+\rho_h)\frac{\beta_v\beta_h N_v^* - \mu_v(\mu_d +\gamma_h) N^*}{\beta_v\beta_h N_v^*(\mu_h+\gamma_h+\rho_h) +\beta_h(\mu_h+\rho_h)(\mu_h+\gamma_h)N^*},\label{is}
\end{equation}
and we see that $I^*$ becomes positive if $\mc R_0>0$, see \eqref{ls}.

To show that also for $\mc R_0>1$, the dynamics of \eqref{mal1} for small values of $\mu_d$ can be approximated uniformly on $[0,\infty)$ by the solution to \eqref{mal3a} for all initial conditions, we establish the global stability of $(N^*,I^*,R^*)$. We see that \eqref{m1}  does not depend on \eqref{m2}, so we can use the Vidyasagar theorem, see, e.g., \cite[Section 5.8.4]{Vid} or \cite[Twierdzenie 5.12]{BanSz}. For this, we consider  \eqref{m2} with $N$ replaced by the fixed equilibrium $N^*$ of \eqref{m1}. Then, at any equilibrium point $(I^*,R^*)$, the Jacobi matrix of \eqref{m2}  is given by $$
\mc J_{(I^*,R^*)} = \left(\begin{array}{cc}
-\frac{\beta_vI^*}{\mu_vN^*+\beta_vI^*}\left(\frac{\beta_hN^*}{N^*}+(\gamma_h +\mu_h)\right)&-I^*\frac{\beta_h\beta_v N_v^*}{N^*(\mu_v  N^*_h+\beta_v I^*)}\\
\gamma_h&-(\rho_h+\mu_h)\end{array}\right).
$$
If $\mc R_0>1$, we have a unique $I^*>0$ and the trace of $\mc J_{(I^*,R^*)}$ is negative, while its determinant is positive, thus $(I^*,R^*)$ is locally asymptotically stable. Next, consider the Dulac function
$$
\phi(I,R) = \frac{\mu_v  N^*_h +\beta_{v}  I}{ I},
$$
well-defined and differentiable in the open first quadrant $\mbb R^2_+$.

 Then we have
\begin{align*}
&\frac{\p}{\p I}(\phi(T,R)f(I,R)) + \frac{\p}{\p R}(\phi(T,R)f(I,R))\\
 &\phantom{x}= \frac{\p}{\p I}\left(\beta_{v}\beta_{h}N_v^*\frac{ N^*- I- R}{ N^*_h} - (\gamma_h+ \mu_{h})(\mu_v  N^* +\alpha_{h}  I)\right)\\
  &\phantom{xx}+ \frac{\p}{\p R} \left(\gamma_h (\mu_v  N^*_h +\beta_{v}  I) - (\rho_h +\mu_{h}) R \frac{\mu_v  N^*_h +\beta_{h}  I}{ I}\right)\\
  &\phantom{x} = - \frac{\beta_{v}\beta_{h}N_v^*}{N^*}-(\gamma_h+ \mu_{h})\beta_{v} -  (\rho_h +\mu_{h}) \frac{\mu_v  N^*_h +\beta_{v}  I}{ I}<0,
\end{align*}
whenever $I>0$, thus there are no closed orbits in $\mbb R^2_+$. For no point $(\mr I,\mr R)\in \mbb R^2_+$ its $\omega$-limit set intersects the semi-axes $\{I=0, R>0\}$ and $\{R=0,I>0\}$. Indeed, in the latter case, the field points inside $\mbb R_+^2$, while in the former case, the whole semi-axis is a trajectory with its $\omega$-limit point being $(0,0)$, which, by invariance, would belong to $\omega(\mr I,\mr R)$. Since, $I'>0$ in the positive vicinity of $(0,0)$, so no internal trajectory can approach $\{I=0, R>0\}$ close to $(0,0)$. Hence this part, belonging to $\omega(\mr I,\mr R)$, must be a part of the internal trajectory. This, however, would violate the uniqueness.   Since the trajectories are bounded, their $\omega$-limit sets are compact and contained in $\mbb R_+^2$. Hence, using the Poincar\'{e}--Bendixon trichotomy, e.g., \cite[Theorem 8.8]{Teschl}, the $\omega$-limit set of any orbit coincides with $(I^*,R^*)$ and hence $(I^*,R^*)$ is globally asymptotically stable. By the Vidyasagar theorem, $(N_h^*,I^*,R^*)$ is globally asymptotically stable.

Let us return to the original problem. Using Theorem \ref{thm1} we can state that, if $\mc R_0>1$, then for a fixed small $\mu_d$, the solutions $$(N_{h,\e}(t),I_{h,\e}(t),R_{h,\e}(t), N_{v,\e}(t),I_{v,\e}(t))$$ to \eqref{mal1} converge as $\e\to 0$ to
$$(\bar N_{h}(t),\bar I_{h}(t),\bar R_{h}(t), \bar N_{v}(t),\bar I_{v}(t)),$$
where $(\bar N_{h}(t),\bar I_{h}(t),\bar R_{h}(t))$ solves \eqref{mal3}, $\bar N_{v} = N^*_v$ and $\bar I_v$ is given by \eqref{qssm} (with $I_h$ and $N_h$ replaced by, respectively, $\bar I_h$ and $\bar N_h$), uniformly on any interval $[t_0,\infty), t_0>0.$ This means that, for small $\e$ , the endemic equilibria of \eqref{mal1e} are close to $(\bar N^*_h, \bar I^*,\bar R^*, N_v^*, \bar I_v^*)$, where $(\bar N^*_h, \bar I^*,\bar R^*)$ is the endemic equilibrium of \eqref{mal3a} and
$$
\bar I^*_v = N_v^*\frac{\sigma_v\beta_{vh}\bar I^*_h}{\mu_v \bar N^*_h +\sigma_v\beta_{vh} \bar I^*_h}.
$$
\subsection{Numerical simulations for Subsection \ref{ss2.3}}
In this section, we shall show some simulations supporting the theoretical results for \eqref{mal1} and \eqref{mal1e}. We use the convention and parameters as in Section \ref{ss22}. For the demography, we use the logistic model with $r=0.1$ and $K_h=74830455.$  We recall that we present simulations for $\mathcal{R}_0>1$ (where we use  $\beta_v=1.5$ and $\beta_h=10$) and $\mathcal{R}_0<1$ (with $\beta_v=0.18$ and $\beta_h=4.4$) on the right.
\begin{figure}[!h]
\begin{center}
\includegraphics[scale=0.4]{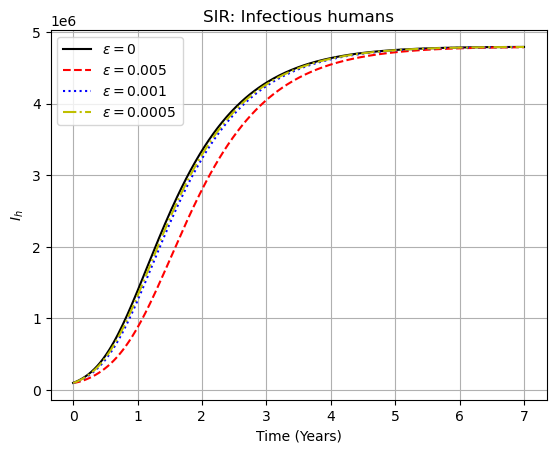}
\includegraphics[scale=0.4]{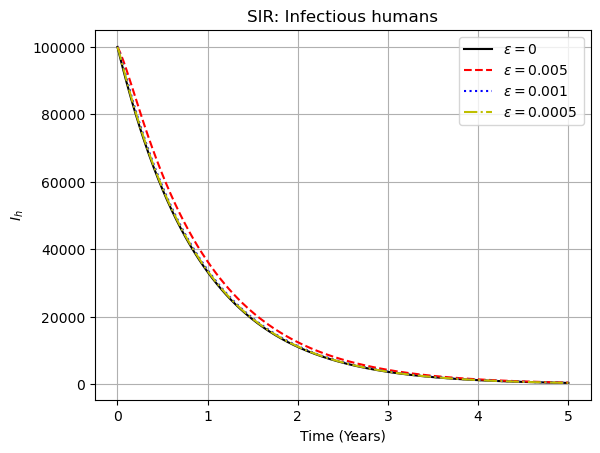}
\end{center}
\caption{These figures show the graphs of the solutions $I_{h,\epsilon}(t)$ of \eqref{mal1e}  for $\epsilon = 0.005,\;0.001\;\text{and}\; 0.0005$ and the solution $I_h(t)$ of the equation  \eqref{m2}, represented by $\epsilon = 0$. }
\label{Ih_L}
\end{figure}

\begin{figure}[!h]
\begin{center}
\includegraphics[scale=0.4]{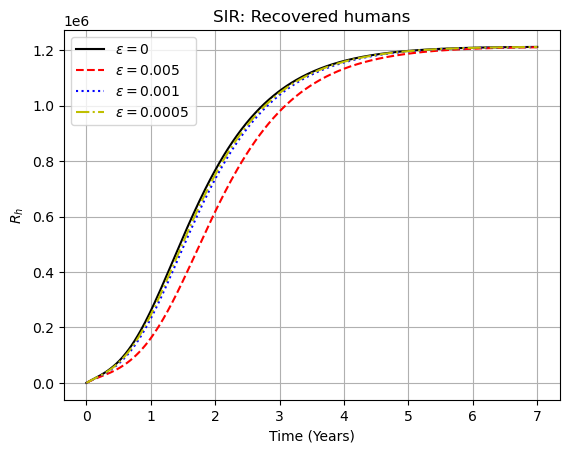}
\includegraphics[scale=0.4]{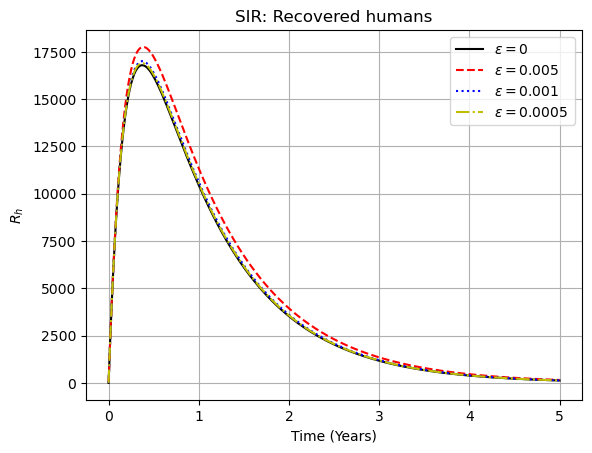}
\end{center}
\caption{Graphs of the solutions $R_{h,\epsilon}(t)$ of \eqref{mal1e}  for $\epsilon = 0.005,0.001$ and $0.0005$ and the solution $R_h(t)$ of \eqref{m2}, represented by $\epsilon = 0$.}
\label{Rh_L}
\end{figure}
\begin{figure}[!h]
\begin{center}
\includegraphics[scale=0.38]{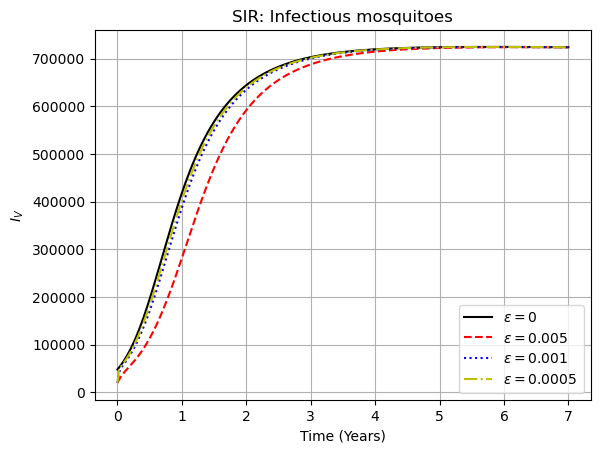}
\includegraphics[scale=0.38]{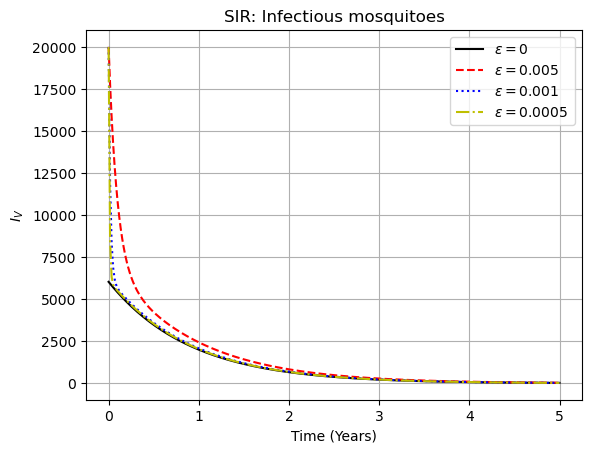}
\end{center}
\caption{Graphs of the solutions $I_{v,\epsilon}(t)$ of \eqref{mal1e}  for $\epsilon = 0.005$, $0.001$ and $0.0005$ and  $\bar{I}_{v,0}(t)$ given in equation \eqref{Iv0}, represented by $\epsilon = 0.$}
\label{Iv_L}
\end{figure}

\appendix
\renewcommand{\theequation}{\thesection.\arabic{equation}}
\setcounter{equation}{0}
\section{Tikhonov theorem}
In this paper, we consider autonomous singularly perturbed systems
\begin{equation}
\begin{split}
\mb u'_{\e,t}&={\mb f}(\mb u_\e,\mb v_\e,\e), \quad \mb u_\e(0) =\mr {\mb u},\\
\e \mb v'_{\e,t} &={\mb g}(\mb u_\e,\mb v_\e,\e),\quad \mb v_\e(0)= \mr{\mb  v},
\end{split}
\label{(iss)}
\end{equation}
where $\mb f$ and $\mb g$ are sufficiently smooth functions defined on open domains of $\mbb R^n\times \mbb R^m\times \mbb R$ acting, respectively, into $\mbb R^n$ and $\mbb R^m$, and  $\e$ is a small positive parameter. We often simultaneously consider the fast version of \eqref{(iss)}, obtained by the rescaling $\tau = \frac{t}{\e}$,
\begin{equation}
\begin{split}
\ti{\mb u}'_{\e,\tau}&=\e {\mb f}(\ti{\mb u}_\e,\ti{\mb v}_\e,\e), \\
\ti{\mb v}'_{\e,\tau} &={\mb g}(\ti{\mb u}_\e,\ti{\mb v}_\e,\e),
\end{split}
\label{CE1a}
\end{equation}
which is equivalent to \eqref{(iss)} for $\e >0$.
 Tikhonov theorem gives conditions ensuring that the solutions $(\mb u_\e(t), \mb v_\e(t))$ of (\ref{(iss)}) converge to
$(\mb{\bar u}(t), \mb{\phi}(t,\bar{\mb u}))$, where $\mb{\bar v} = \mb\phi(\mb u)$ is the solution to the  equation
\begin{equation}
0 =  \mb g(\mb u,\mb v,0),
\label{deg1}
\end{equation}
called the \textit{quasi-steady state (QSS)}, and $\mb{\bar u}(t)$ solves the \textit{reduced equation}
\begin{equation}
{\mb u}'_{,t}  =  \mb f(\mb u,\mb{\phi}(\mb u),0), \qquad \mb u(0)  =  \,\mr{\mb u},
\label{deg2}
\end{equation}
obtained from the first equation of (\ref{(iss)}) by substituting the unknown $\mb v$ by the known quasi steady state $\bar{\mb  v}$ (dependent on $\mb u$). We assume that \eqref{deg2} is uniquely solvable on some interval $[0,T].$

\noindent
 \textbf{Main assumptions: }
  \begin{description}
   \item {a)} the quasi-steady states are isolated in some subset $\bar{\mc U}$ of the domain of \eqref{(iss)};
   \item {b)} for each fixed  $\mb u$, the quasi-steady state solution $\mb \phi(\mb u)$ of (\ref{deg1}) is a uniformly asymptotically stable equilibrium of \begin{equation}\label{auxil'}
\tilde{\mb v}'_{,\tau} = \mb g(\mb u,\tilde{\mb v},0);
\end{equation}
  \item {c)} $\bar{\mb  u}(t) \in \mc U$ for $t\in [0,T]$ provided $\,\mr {\mb u} \in \bar{\mc U}$;
  \item {d)} solutions to
  \begin{equation}
  \hat{\mb v}'_{,\tau} = \mb g(\mr {\mb u},\hat{\mb v},0), \quad \hat{\mb v}(0) = \mr {\mb v}
  \label{ileq}
  \end{equation}
  converge to $\bar{\mb v}( \,\mr {\mb u})$ as $\tau\to \infty$.
   \end{description}
Then, the following theorem is true.\setcounter{theorem}{0}
\begin{theorem}\label{tikhonovth'}
Let the above assumptions be satisfied. Then
there exists $\varepsilon_0 >0$ such that for any $\varepsilon\in\, (\,0,\varepsilon_0 ]$ there
exists a unique solution
$(\mb u_{\varepsilon} (t),\mb v_{\varepsilon} (t))$ of Problem (\ref{(iss)}) on $[0,T ]$ and
\begin{subequations}\label{tikhonovcon'}
\begin{align}
\lim\limits_{\varepsilon\to 0} \mb u_{\varepsilon} (t) & =\bar{\mb u}(t),\qquad t\in [0,T ],\label{16a}\\
\lim\limits_{\varepsilon\to  0} \mb v_{\varepsilon} (t) & = \bar{\mb v}(t),\label{16b}
\qquad  t\in \, (\, 0,T ]\,,
\end{align}
\end{subequations}
where $\bar{\mb u}(t)$ is the solution of  (\ref{deg2}) and $\bar{\mb v}(t)= \mb{\phi}( \bar{\mb u}(t))$ is the solution of (\ref{deg1}).
\end{theorem}

 The convergence in (\ref{16a}) is uniform in  $t\in [0,T]$, but in (\ref{16b}) it is uniform only in each  $[\zeta  ,T]$, $\zeta  >0$. This is the so-called initial layer effect. We can eliminate this lack of convergence by adding the initial layer corrector, the solution to \eqref{ileq}.

As noted in the Introduction, the assumptions of the Tikhonov theorem, even if they are satisfied on the whole space, do not suffice to extend \eqref{tikhonovcon'} to the interval $(0,\infty)$.
%\begin{example}\label{exA1}
%Consider the family of the Cauchy problems
%\begin{equation}
%u_\e' = \e u_\e^2,\quad u_\e(0)=1.
%\label{regpert0}
%\end{equation}
%Then $u_0(t) = 1, t\geq 0,$ and
%$$
%u_\e(t) = \frac{1}{1-\e t}, \quad 0\leq t< \frac{1}{\e}.
%$$
%It is clear that for $\e>0,$ $u_\e(t)$ blows up as $t\to \frac{1}{\e}^-$ but, for any $T<\infty,$ the solution $u_\e(t)$ is defined on $[0,T]$ for $\e<\e_0 = \frac{1}{T}$ and satisfies
%$$
%\lim\limits_{\e\to 0^+} u_\e(t) = u_0(t)
%$$
%uniformly on $[0,T]$.\end{example}

We shall give a generalization of Theorem \ref{tikhonovth'} valid uniformly on unbounded time intervals. Our result also includes correction terms allowing for such an approximation to be valid with $O(\e^2)$ accuracy. For this, we need an algorithm
to construct such corrections, which is provided in the next section.

\section{ Approximation of the slow manifold and the Chapman--Enskog procedure}\label{ChE}
\subsection{The Chapman--Enskog procedure}
\setcounter{equation}{0}
Considering this paper's applications, we assume that \eqref{(iss)} is autonomous and that $\mb g$ is independent of $\e$. Then, we look for the expansion $({\mb u}_\e,{\mb v}_\e) = (\bar {\mb u}_\e, \bar {\mb v}_\e)+O(\e^2)= (\bar {\mb u}_\e, \bar {\mb v}_0+\e \bar {\mb v}_1) +O(\e^2)$, so that  \eqref{(iss)} takes the form
\begin{subequations}\label{CE2}
\begin{align}
&\bar {\mb u}'_{\e,t}= {\mb f}(\bar {\mb u}_\e,\bar {\mb v}_0)+\e({\mb f}_{,\mb v}(\bar {\mb u}_\e,\bar {\mb v}_0,0)\bar {\mb v}_1+\mb f_{,\e}(\bar {\mb u}_\e,\bar {\mb v}_0,0))+O(\e^2), \label{CE2a}\\
&\e (\bar {\mb v}'_{0,t}+\e\bar {\mb v}'_{1,t}) ={\mb g}(\bar {\mb u}_\e,\bar {\mb v}_0+\e\bar {\mb v}_1) + O(\e^2)\label{CE2b},
\end{align}
\end{subequations}
where $\phantom{x}_{,\mb v}$ denotes the Jacobi matrix of the respective function with respect to $\mb v$.
%\begin{remark}
%At the $O(1)$ level, we have the approximation $({\mb u}_\e,{\mb v}_\e) = (\bar {\mb u}, \bar {\mb v})+O(\e)$ given by the standard Tikhonov reduced system, that is, in the current notation,
%\begin{equation}
%\begin{split}
 %\bar {\mb u}_{,t}&={\mb f}%(\bar {\mb u},\mb\phi(\bar {\mb u}),0), \\
 %\bar {\mb u}(0)&=\mr{\mb u},
 %\end{split}
%label{zerolev}
%\end{equation}
%and
%\begin{align}
%\bar {\mb v} (t)&:= {\mb \phi}(\bar{\mb u}(t)),
%\label{CE4}
%\end{align}
%where $\bar {\mb v} $ satisfies
%\begin{equation}
%{\mb g}(\bar{\mb u},\bar {\mb v} ) = {\mb g}(\bar{\mb u},{\mb \phi}(\bar{\mb u})) \equiv 0.
%\label{algeq}
%\end{equation}\end{remark}
In what follows, we shall provide a sketch of the derivation of the higher-order corrections, referring the reader to \cite{BanRG} for details. Expanding \eqref{CE2b} around $(\bar{\mb u}_\e,\bar{\mb v}_0)$ with yet unknown  $\bar{\mb u}_\e,$ we find $\bar{\mb v}_0$ satisfying
\begin{equation}
\mb g(\bar{\mb u}_\e, \bar{\mb v}_0)= \mb g(\bar{\mb u}_\e, \mb \phi(\bar{\mb u}_\e)) \equiv 0,
\label{newv}
\end{equation}
which is the same as \eqref{deg1} (remember $\mb \phi$ is independent of $\e$) and hence
$\mb v_1$ can be obtained from \begin{equation}\label{newv1}
 {\mb g}_{,\mb v}( \bar{\mb u}_\e,\bar {\mb v}_0)\bar {\mb v}_1=\bar{\mb v}'_{0,t}.
\end{equation}
 If ${\mb g}_{,\mb v}$ is nonsingular in some tubular neighbourhood of the solution to \eqref{deg1}, for $\bar{\mb u}_\e$  close to $\bar{\mb u}$
\begin{align}
\bar {\mb v}_{1} &= {\mb g}^{-1}_{,{\mb v}}(\bar{\mb u}_\e,\mb \phi(\bar{\mb u}_\e))\bar {\mb v}'_{0,t}
= -[{\mb g}^{-1}_{,{\mb v}}(\bar{\mb u}_\e,\mb \phi(\bar{\mb u}_\e))]^2{\mb g}_{,\mb u}(\bar{\mb u}_\e,\mb \phi(\bar{\mb u}_\e)) \bar {\mb u}'_{\e,t}.     \label{CE5}\end{align}
  Thus, substituting \eqref{CE5} into \eqref{CE2a}
and discarding the $O(\e^2)$ terms to provide an approximate closure, we obtain that the  approximation of $({\mb u}_\e,{\mb v}_\e)$ at the $\e$ level is given by $(\bar{\mb u}_\e(t),\bar{\mb v}_\e(t))=(\bar{\mb u}_\e(t),\bar{\mb v}_0(t)+ \e\bar{\mb v}_1(t))$, where
\begin{equation}\label{CE5a}
\begin{split}
 \bar{\mb u}'_{\e,t}&={\mb f}(\bar{\mb u}_\e,{\mb \phi}(\bar{\mb u}_\e),0)\\
 &\phantom{x}- \e{\mb f}_{,\mb v}(\bar{\mb u}_\e,{\mb \phi}(\bar{\mb u}_\e),0)[{\mb g}^{-1}_{{,\mb v}}(\bar{\mb u}_\e,{\mb \phi}(\bar{\mb u}_\e))]^2{\mb g}_{,\mb u}(\bar{\mb u}_\e,{\mb \phi}(\bar{\mb u}_\e)) {\mb f}(\bar{\mb u}_\e,{\mb \phi}(\bar{\mb u}_\e))\\
 &\phantom{x}+\e{\mb f}_{,\e}(\bar{\mb u}_\e,{\mb \phi}(\bar{\mb u}_\e),0),
 \end{split}
 \end{equation}
 while
 \begin{equation}
 \bar{\mb v}_\e=\bar{\mb v}_{0}+\e \bar{\mb v}_{1}= {\mb \phi}({\bar{\mb u}_\e}) -\e [{\mb g}^{-1}_{,{\mb v}}({\bar{\mb u}_\e},{\mb \phi}({\bar{\mb u}_\e}))]^2{\mb g}_{,\mb u}({\bar{\mb u}_\e},{\mb \phi}({\bar{\mb u}_\e})) {\mb f}({\bar{\mb u}_\e},{\mb \phi}({\bar{\mb u}_\e}),0).
\label{CE5b}
\end{equation}
The above terms are referred to as the bulk part of the expansion. Since, however, $\bar{\mb v}_\e$ is determined by $\bar{\mb u}_\e$ through an algebraic relation, we obtain the initial layer effect mentioned earlier.   To address this problem, we introduce the initial layer in the standard way, that is, we look for the approximation of the solution in the form
\begin{equation}({\mb u}_\e,{\mb v}_\e) = (\bar {\mb u}_\e + \e \ti {\mb u}_1 +O(\e^2), \bar {\mb v}_{0}+\e \bar {\mb v}_{1} + \ti {\mb v}_0 + \e\ti {\mb v}_1 + O(\e^2)).\label{oe2approx}\end{equation}
We note the absence of $\ti {\mb u}_0$ since the approximation for $\mb u_\e$ may be constructed to satisfy the original initial data. To find the initial layer, we look for the approximation of \eqref{CE1a},
\begin{align}
&\bar {\mb u}'_{\e,t}(\e\tau) + \ti {\mb u}'_{1,\tau}(\tau)= {\mb f}(\bar {\mb u}_\e(\e\tau) + \e \ti {\mb u}_1(\tau),\bar {\mb v}_0(\e\tau)+\e\bar {\mb v}_1(\e\tau)+ \ti {\mb v}_0(\tau) + \e\ti {\mb v}_1(\tau),\e), \nn\\
&\e\bar {\mb v}'_{0,t}(\e\tau)+\e^2\bar {\mb v}'_{1,t}(\e\tau) + \ti {\mb v}'_{0,\tau}(\tau) + \e\ti {\mb v}'_{1,\tau}(\tau) \label{CE2il}\\
&\phantom{xxxxxxxxxxxxx}={\mb g}(\bar {\mb u}_\e(\e\tau) + \e \ti {\mb u}_1(\tau),\bar {\mb v}_0(\e\tau)+\e\bar {\mb v}_1(\e\tau)+ \ti {\mb v}_0(\tau) + \e\ti {\mb v}_1(\tau) ),\nn
\end{align}
where we discarded $O(\e^2)$ terms. At the zeroth level of approximation, the second equation gives
\begin{equation}\begin{split}
\ti {\mb v}'_{0,\tau}(\tau) &= {\mb g}(\mathring {\mb u}, {\mb \phi}(\mathring {\mb u})+ \ti {\mb v}_0(\tau)),\\
\ti {\mb v}_0(0)&= \mathring {\mb v} -{\mb \phi}(\mathring {\mb u}),
\end{split}
\label{tiv0}
\end{equation}
where $\ti {\mb v}_0(\tau)$ converges to 0 as $\tau\to \infty$ by Main assumption b).
Next, setting $\e=0$ in the first equation of  \eqref{CE2il} and  using \eqref{CE5a}, also with $\e=0$,
$$
\ti {\mb u}'_{1,\tau} (\tau) = {\mb f}(\mr {\mb u}, {\mb \phi}(\mr {\mb u})+\ti{\mb v}_0(\tau),0) - {\mb f}(\mr {\mb u}, {\mb \phi}(\mr {\mb u}),0).
$$
Since $\ti {\mb u}_1$ should vanish at infinity, the solution must be given by
\begin{equation}
\ti {\mb u}_{1}(\tau) = -\cl{\tau}{\infty}({\mb f}(\mr {\mb u}, {\mb \phi}(\mr {\mb u})+\ti{\mb v}_0(\tau),0) - {\mb f}(\mr {\mb u}, {\mb \phi}(\mr {\mb u}),0))ds
\label{tiu1}
\end{equation}
and thus its initial value for $\ti {\mb u}_{1}$ is predetermined as
\begin{equation}
\ti {\mb u}_{1}(0) = -\cl{0}{\infty}({\mb f}(\mr {\mb u}, {\mb \phi}(\mr {\mb u})+\ti{\mb v}_0(\tau),0) - {\mb f}(\mr {\mb u}, {\mb \phi}(\mr {\mb u}),0))ds.
\label{tiu1iv}
\end{equation}
To balance this, by \eqref{oe2approx},  we consider \eqref{CE5a} with the following  correction to $\mr {\mb u},$
\begin{equation}
\bar {\mb u}_\e(0) = \mr {\mb u} + \e \cl{0}{\infty}({\mb f}(\mr {\mb u}, {\mb \phi}(\mr {\mb u})+\ti{\mb v}_0(\tau),0) - {\mb f}(\mr {\mb u}, {\mb \phi}(\mr {\mb u}),0))ds = \mr {\mb u} -\e\ti{\mb u}_{1}(0).
\label{baruic}
\end{equation}
We refrain from giving the formula for $\ti{\mb v}_1$ due to its length. %Finally,  tedious calculations give the term $\ti {\mb v}_1$ in \eqref{oe2approx} as the solution of
%\begin{align*}
%&\ti{\mb v}_{1,\tau}(\tau) = {\mb g}_{,\mb v} (\mr {\mb u}, {\mb \phi}(\mr {\mb u})+\ti{\mb v}_0(\tau) )\ti {\mb v}_1(\tau)
%+   {\mb g}_{,\mb u}(\mr {\mb u}, {\mb \phi}(\mr {\mb u})+\ti{\mb v}_0(\tau) )\ti {\mb u}_1(\tau) \nn\\&\phantom{x}+ ({\mb g}_{,\mb u}(\mr {\mb u}, {\mb \phi}(\mr {\mb u})+\ti{\mb v}_0(\tau) )-{\mb g}_{,\mb u}(\mr {\mb u}, {\mb \phi}(\mr {\mb u})))(\bar {\mb u}_{,t}(0)\tau -\ti{\mb u}_1(0))\nn\\ &\phantom{x}+ ({\mb g}_{,\mb v} (\mr {\mb u}, {\mb \phi}(\mr {\mb u})+\ti{\mb v}_0(\tau) )-{\mb g}_{,\mb v}(\mr {\mb u}, {\mb \phi}(\mr {\mb u})))
%(\bar {\mb v}_{,t}(0)\tau - \mb{\phi}_{,\mb u}(\mb{\mr u},0)\ti{\mb u}_1(0)+ \bar {\mb v}_1(0) ), \end{align*}
%with the initial condition
%$$
%\ti {\mb v}_1(0) = \mb\phi_{,\mb u}(\mr{\mb u})\mb{\ti u}_1(0)- \bar {\mb v}_1(0) =
% \mb\phi_{,\mb u}(\mr{\mb u})\mb{\ti u}_1(0)+{\mb g}^{-1}_{,{\mb v}}(\mr {\mb u},\mb\phi(\mr{\mb u}))\bar {\mb v}_{0,t}.\label{v1ic}
%$$

For the validity of the above calculations and subsequent error estimates, we need to adopt several assumptions. As they are quite technical, here
we shall provide a brief description of the most salient ones, referring the reader to \cite{BanRG, BanViet} for details. In general, we assume that all assumptions for the validity of Theorem \ref{tikhonovth'} are satisfied on $\mbb R^n\times \mbb R^m$ and \eqref{deg1}
admits an isolated solution $\bar{\mb v} = \mb\phi(\mb u)$ for any ${\mb u}\in \mbb R^n$ such that \eqref{deg2} has a unique solution $\bar{\mb u}(t)$ on $[0,\infty)$,  bounded together with its derivative. We strengthen assumption b) by assuming
\begin{equation}
\sup\limits_{t\in [0,\infty)} s({\mb g}_{,\mb v}(\bar{\mb u}(t), \mb{\phi}(\bar{\mb u}(t)))=: -\kappa<0,
\label{mug}
\end{equation}
for some $\kappa>0$, where for a matrix $A$, $s(A)$ denotes its spectral bound, that is, the maximum of the real parts of its eigenvalues. Regularity assumptions ensure that \eqref{mug} is satisfied, with possibly different constant, in some tubular neighbourhood of the trajectory of $(\bar{\mb u}(t), \mb{\phi}(\bar{\mb u}(t))$. This assumption means that the QSS (also called the  \textit{slow manifold}), \begin{equation}\mc M = \{({\mb u},{\mb v});\; {\mb v} = \mb \phi({\mb u})\},\label{sm}\end{equation} consists of hyperbolic attractive equilibria of the fast dynamics $\mb v_{,\tau} = \mb g(\mb u,\mb v)$, see \eqref{CE1a}. Further, we assume that $\mr{\mb v}$ belongs to the basin of attraction of $\mc M$, that is, the solution $\ti{\mb v}_0$ to \eqref{tiv0} exponentially converges to $0$ as $\tau\to \infty$, see \cite[Remark 1]{BanRG}. Also, $\ti {\mb u}_1$ and $\ti {\mb v}_1$ have the same property, see \cite[Lemma 1]{BanRG}.

 The crucial assumption for the validity of the approximation on $[0,\infty)$, \cite{Hop, BanRG},  is the requirement that the Jacobi matrix
$$
\mc J_{\mb f}(\bar{\mb u}(t)):=\mb f_{,\mb u}(\bar{\mb u}(t), \mb{\phi}(\bar{\mb u}(t)),0) + \mb f_{,\mb v}(\bar{\mb u}(t), \mb{\phi}(\bar{\mb u}(t)),0){\mb\phi}_{,\mb u}(\bar{\mb u}(t)), \quad t\in [0,\infty),
$$
has the exponential dichotomy property (or is uniformly exponentially stable), see \cite[Chapter III]{Cop}. For  this paper's applications, it will suffice   to know, \cite[Remark 3]{BanViet}, that if $\bar{\mb u}(t) \to \bar{\mb u}^*$ as $t\to \infty$ and the spectral bound of $\mc J_{\mb f}(\bar{\mb u}^*)$ satisfies
\begin{equation}
s(\mc J_{\mb f}(\bar{\mb u}^*)) = s(\mb f_{,\mb u}(\bar{\mb u}^*, \mb{\phi}(\bar{\mb u}^*),0) + \mb f_{,\mb v}(\bar{\mb u}^*, \mb{\phi}(\bar{\mb u}^*),0){\mb\phi}_{,\mb u}(\bar{\mb u}^*))<0,
\label{A5}
\end{equation}
(that is, $\bar{\mb u}^*$ is exponentially stable) then  $\mc J_{\mb f}(t)$ also has the exponential dichotomy property. We have, \cite[Theorem 3]{BanRG},
 \begin{theorem} \label{thm1} Under the above assumptions the Chapman-Enskog approximation errors
 \begin{equation}
 \begin{split}
\mb \zeta_\e(t) &= \mb u_\e(t) - \bar{\mb u}_\e(t) -  \e\ti{\mb u}_1(\tau),\\
\mb \eta_\e(t) &= \mb v_\e(t) - \mb \phi(\bar{\mb u}_\e(t)) - \ti{\mb v}_0(\tau) - \e\bar{\mb v}_{1}(t) - \e\ti{\mb v}_1(\tau),
\end{split}
\label{ereat1}
\end{equation}
satisfy
\begin{equation*}
\mb \zeta_\e(t) = O(\e^2), \qquad \mb \eta_\e(t) = O(\e^2)
%\label{zeetaest1a}
\end{equation*}
as $\e\to 0$ uniformly on $[0,\infty)$.
\end{theorem}
We note that by including $O(\e)$ terms in \eqref{ereat1} in the error terms, the zeroth order terms give $O(\e)$ error uniform on $[0,\infty)$.

\subsection{ Fenichel's geometric singular perturbation theory, the Chapman--Enskog method  and the group renormalization method}
A simpler version of Theorem \ref{thm1} was proved in \cite{MarCz} where, however, the construction of the asymptotic expansion was carried out by a more cumbersome renormalization group method; also, the authors did not construct the first order initial layer corrector $\ti{\mb v}_1,$  which resulted in $O(\e^2)$ error being valid only on $[t_0,\infty)$ for any $t_0>0$. To compare the renormalization group and Chapman--Enskog methods, we refer the reader to \cite{BanRG}.

On the other hand, the Chapman--Enskog expansion can be derived from Fenichel's formulation of the Tikhonov theory. Roughly speaking, see \cite{Kuehn}, Fenichel's theory states that if $\mc M_0$ is a compact submanifold of the slow manifold $\mc M$ satisfying the assumptions of the Tikhonov theorem, then for each sufficiently small $\e$ there is a manifold $\mc M_\e$ diffeomorphic to $\mc M_0,$ which is locally invariant with respect to the flow generated by \eqref{(iss)} and this flow converges to the slow flow on $\mc M_0$ as $\e\to 0$. Following e.g., \cite{RashVent}, we assume that locally $\mc M_\e$ is given as the graph $\mb v = \mb \phi_\e(\mb u)$, with $\mb\phi_0=\mb\phi$. If it is invariant with respect to \eqref{(iss)}, then for any solution $t\mapsto (\mb u_\e(t),\mb v_\e(t))$ we must have $\mb v_\e(t) \equiv \mb \phi_\e(\mb u_\e(t))$ so that
$$
\mb v_{\e,t}(t) \equiv \mb \phi_{\e,\mb u}(\mb u_\e(t))\mb u_{\e,t}(t),
$$
and hence, using \eqref{(iss)},
$$
\mb g(\mb u_\e(t),\mb v_\e(t)) = \e \mb \phi_{\e,\mb u}(\mb u_\e(t))\mb f(\mb u_\e(t),\mb v_\e(t)).
$$
Expanding $\mb v_\e = \mb v_0 +\e \mb v_1+\ldots$, we get
$$
\mb g(\mb u_\e,\mb v_0)+ \e\mb g_{,\mb v}(\mb u_\e,\mb v_0)\mb v_1 = \e \mb \phi_{,\mb u}(\mb u_\e) \mb f(\mb u_\e,\mb v_0)) +O(\e^2)
$$
from where
$$
\mb v_0 = \mb\phi(\mb u_\e).
$$
Next, since $\mb g(\mb u,\mb \phi(\mb u)) \equiv 0$, we have $$\mb g_{,\mb u}(\mb u,\mb \phi(\mb u)) + \mb g_{,\mb v} (\mb u,\mb \phi(\mb u))\mb \phi_{,\mb u}(\mb u)\equiv 0$$
and thus
$$
\mb v_1 = -[\mb g^{-1}_{,\mb v} (\mb u_\e,\mb \phi(\mb u_\e))]^2\mb g_{,\mb u}(\mb u_\e,\mb \phi(\mb u_\e))\mb f(\mb u_\e,\mb \phi(\mb u_\e)),
$$
which corresponds to \eqref{CE5b}. In other words, the Chapman-Enskog bulk expansion is equivalent to considering the slow equation of \eqref{(iss)} on the approximation of the invariant manifold $\mc M_\e$ of \eqref{(iss)} up to $O(\e)$ terms. Note, however, that the Fenichel theorem provides neither the initial layer terms nor long-term estimates (in slow time).

%\bibliographystyle{abbrv}
%\bibliography{Dropbox/FCBook/BLLbook/BLL_Bk}
%\bibliography{ALrefsingpert}

\end{document}